\documentclass[10pt,journal,compsoc]{IEEEtran}
\ifCLASSOPTIONcompsoc
  \usepackage[nocompress]{cite}
\else
  \usepackage{cite}
\fi
\ifCLASSINFOpdf
\else
\fi
\ifCLASSOPTIONcompsoc
  \usepackage[caption=false,font=footnotesize,labelfont=sf,textfont=sf]{subfig}
\else
  \usepackage[caption=false,font=footnotesize]{subfig}
\fi

\ifCLASSINFOpdf                       
  \pdfoutput=1\relax                   
  \usepackage{graphicx}                
  \DeclareGraphicsExtensions{.pdf,.png,.jpg,.jpeg} 
\else
  \usepackage{graphicx}                
  \DeclareGraphicsExtensions{.eps}     
\fi%

\graphicspath{{figs/}{figures/}{pictures/}{images/}{./}} 

\usepackage{microtype}                 
\PassOptionsToPackage{warn}{textcomp}  
\usepackage{textcomp}                  
\usepackage{mathptmx}                  
\usepackage{times}                     
\usepackage{cite}                      
\usepackage{tabu}                      
\usepackage{booktabs}                  
\usepackage{xspace}
\usepackage{color}
\usepackage{amsmath}
\usepackage{amssymb}
\usepackage{enumitem}
\usepackage{balance}
\usepackage{wrapfig}
\usepackage{amsthm}
\usepackage[hyphens,spaces,obeyspaces]{url}

\newcommand{\hidecomment}[1]{}

\newcommand{\reffig}[1]{{Fig.~\ref{#1}}}

\newcommand{\ie}{\emph{i.e.}\xspace}
\newcommand{\eg}{\emph{e.g.}\xspace}
\newcommand{\etc}{\emph{etc}\xspace}

\newcommand{\myparagraph}[1]{\vspace{0.05cm}\noindent \textbf{#1}}

\newcommand{\denselist}{\itemsep 0pt\parsep=1pt\partopsep 0pt}
\newcommand{\decsp}{\vspace{-0.1in}}

\newtheorem{theorem}{Theorem}

\begin{document}
%
\title{Shadow Accrual Maps: Efficient Accumulation of City-Scale Shadows Over Time}

%
%
%
%

\author{
Fabio~Miranda$^*$,
Harish~Doraiswamy$^*$,~\IEEEmembership{Member,~IEEE},
Marcos~Lage,
Luc~Wilson,
Mondrian~Hsieh 
and Cl\'{a}udio~T.~Silva,~\IEEEmembership{Fellow,~IEEE}

\thanks{$^*$~These authors contributed equally to this work.}
\IEEEcompsocitemizethanks{\IEEEcompsocthanksitem F. Miranda, H. Doraiswamy and C. Silva are with New York University.
E-mail: \{fmiranda, harishd, csilva\}@nyu.edu.%
\IEEEcompsocthanksitem M. Large is with Universidade Federal Fluminense.
E-mail: mlage@ic.uff.br.%
\IEEEcompsocthanksitem L. Wilson and M. Hsieh are with Kohn Pedersen Fox Associates PC.
E-mail: \{lwilson, mhsieh\}@kpf.com.}%
\thanks{Manuscript received July 28, 2017; revised July 28, 2017.}}

\markboth{IEEE Transactions on Visualization and Computer Graphics,~Vol.~25,~No.~3,~March~2019}%
{Shell \MakeLowercase{\textit{Miranda, Doraiswamy, et al.}}: Shadow Accrual Maps: Efficient Accumulation of City-Scale Shadows over Time}
%



\IEEEtitleabstractindextext{%
\begin{abstract}
Large scale shadows from buildings in a city play an important role in determining the environmental quality of public spaces. They can be both beneficial, such as for pedestrians during summer, and detrimental, by impacting vegetation and by blocking direct sunlight. Determining the effects of shadows requires the accumulation of shadows over time across different periods in a year. In this paper, we propose a simple yet efficient class of approach that uses the properties of sun movement to track the changing position of shadows within a fixed time interval. We use this approach to extend two commonly used shadow techniques, shadow maps and ray tracing, and demonstrate the efficiency of our approach.
Our technique is used to develop an interactive visual analysis system, Shadow Profiler, targeted at city planners and architects that allows them to test the impact of shadows for different development scenarios. We validate the usefulness of this system through case studies set in Manhattan, a dense borough of New York City.
\end{abstract}

}

\maketitle

\IEEEdisplaynontitleabstractindextext

%
\IEEEpeerreviewmaketitle


\section{Introduction}
\label{sec:intro}
%
A rapid increase in the urbanization of the world's population~\cite{unurban} has resulted in the need for cities to \textit{densify} to equitably meet the rising housing demands while still maintaining the \textit{environmental quality} of public spaces such as streets and parks.
%
%
A key quantity that plays a crucial role in defining this quality is the impact of shadows from buildings. In particular, shadows can potentially infringe on the ``right to light" of other citizens in the community through the occlusion of direct sunlight by shading public spaces. This can not only inhibit vegetation growth but also reduce solar energy potential. On the other hand, shadows can also be beneficial by reducing the urban heat island effect that paved surfaces create, or by providing a comfortable environment for park goers.
It is therefore important to maintain a balance in the amount of shadows cast with the development of a city.

This requires extensive analysis to be performed to allow for amiable negotiations between the various stakeholders including the city council, the urban designers and developers, and other government agencies.
However, in practice there is little analysis of cast shadows being done to test the impact of new development primarily due to the non-availability of the necessary tools. 
While cities do perform shadow analysis, the time and cost involved limits it mostly to a small and discrete set of times and very specific instances (\eg see \cite{nyc-shadows,sfo-shadows}).
It is therefore crucial that efficient and interactive tools are within purview of stakeholders since this allows for 1)~a more comprehensive analysis; and 2)~democratization of the planning process -- making these results accessible and allowing the general public to visualize various scenarios will also help them contribute to the dialogue around new policies. Interactivity in such analysis also helps architects and developers quickly iterate over several possible designs when working on a project.

The computation of shadows is one of the most popular topics of research in the computer graphics domain. Due to its importance in realistic rendering, several techniques have been proposed for computing shadows both in real-time, as well as offline~\cite{Eisemann:2011:RS:2049989,woo2012shadow}. These techniques were designed to support a variety of scenarios involving the scene as well as lighting options. All of these techniques, however, typically consider only scenes with a fixed set of light(s).
%
%
Orthogonal to these techniques, we are interested in quantifying shadows over multiple time periods of interest. In particular, we are interested in quantifying the amount of time a given location is in shadow over the given time periods. This requires the \textit{accumulation} of shadows involving different time periods of varying lengths. In addition, we are also interested in measuring how proposed developments can affect the accumulation of shadows in its neighborhood.

%
While none of the existing techniques directly support the accumulation of shadows over time, they can still be extended to accumulate shadows. 
The most straightforward approach, also followed by currently used tools~\cite{sunroof,revit}, is to explicitly identify shadows for each time step of a given interval.
The direction of sun light depends not only on a city's location, but also on the time of the year. This makes the shape of shadow highly dependent on the day and time, requiring shadows to be computed for potentially several thousand time steps depending on the temporal resolution. Combining this with the scale of a city, which is typically spread over a wide area consisting of thousands of buildings, makes it computationally expensive to be performed at a suitably high resolution. More importantly, this increased complexity also hampers the \textit{interactivity} of the analysis pipeline.
For example, to accumulate shadows in a 3-hour period over a week at a time resolution of 1~minute would require shadows to be computed for 1260 light positions. Doing this for larger time periods spanning several months (such as summer or winter) at high resolutions is not practical.\\

Another option is to pre-compute all possible shadows, or appropriate data structures such as shadow maps or shadow volumes, and use this for the analyses. 
However, as mentioned above, we are interested in exploring the impact of new constructions with respect to the shadow. This requires interactively testing and iterating over several designs, and any pre-computation based approach will not be able to efficiently handle such scenarios, since the data structures will have to be recomputed based on the new set of conditions.

A third option is to model the given time interval as a set of directional lights. This approach requires the ability to support several thousand to tens of thousand light sources. However, existing techniques are catered towards only a small number of light sources, making such an extension non-trivial.

\myparagraph{Contributions.}
In this paper, motivated by problems faced by architects and city planners, we take the first step at interactively accumulating shadows over time.
We first define two shadow accumulation quantities to effectively quantify the accumulation of shadows over time. We then propose a simple approach to efficiently accumulate shadows over time, and which can be used to speed up existing shadow techniques. Our shadow algorithms are then used to develop an interactive visual exploration system targeted at city planners and architects. Our contributions are summarized as follows:
\begin{itemize}[leftmargin=*]\denselist
\item We propose a simple approach to accumulate shadows that implicitly tracks the movement of shadows. It is accomplished by taking advantage of the properties of sun movement within short time intervals. 
\item Our proposed approach is used to extend two common shadow techniques -- shadow maps and ray tracing, to efficiently accumulate shadows. In particular, we present \textit{shadow accrual maps} which extend standard shadow maps to accumulate shadows over time, and \textit{inverse accrual maps} which use ray tracing to identify shadow movement, thus allowing shadows to be accumulated by simply drawing a set of lines.
\item Making use of the coherency in the sun directions, we design optimizations that allow for the interactive accumulation of shadows over large time intervals. 
\item We show experimental evaluation demonstrating the accuracy and performance of our technique. We show that on average, the shadow accumulation using shadow accrual maps performs around an order of magnitude faster than a naive baseline.
\item We develop \textit{Shadow Profiler}, an interactive visual analysis system targeted at city planners and architects to explore shadows in a city, and test the impact of multiple scenarios. 
\item We show the utility of Shadow Profiler through case studies  set in Manhattan, New York City. The case studies evaluate accumulated shadows over Central Park to study a set of \textit{supertall} towers currently under construction that has generated intense public debate.
\end{itemize}

\section{Related Work}
\label{sec:related}

We briefly survey existing literature from three categories: visual analytics in the context of cities, the study of shadows in urban design, and shadow computation techniques from computer graphics.

\myparagraph{Urban visual analytics.}
Multiple visual analytics systems have been proposed to interactively explore and analyze urban data~\cite{glander2009abstract}. 
These systems are primarily designed to analyze urban data generated by the urban environment. 
For example, there are individual visual analytics systems in transportation and mobility~\cite{ferreira2013visual,zeng2014visualizing,wang2013visual,andrienko2008spatio}, air pollution~\cite{qu2007visual},
real-estate ownership~\cite{hoang2014towards,sun2013web} and public utility service problems~\cite{zhang2014visual}.
There are also tools developed to multiple urban data sets~\cite{chang2007visualization,urbane}.
We refer the reader to Zheng~et~al.~\cite{urban-vs-survey} for a comprehensive survey on visual analytics approaches in urban computing.

Recently, several software platforms have also emerged that aim to use urban data sets together with city geometry to help inform the decision making process in the development of cites. They are aimed at a range of stakeholders (architects, city planners, developers, and the general public), such as Place I Live~\cite{placeilive}, Transitmix~\cite{transitmix}, Flux~\cite{fluxmetro}, ViziCities~\cite{vizicities}, ArcGIS~\cite{johnston2001using}, Urbane~\cite{urbane}, and Vis-A-Ware~\cite{Ortner2017}.
Of these only Urbane, Vis-A-Ware and ArcGIS support computing impact based on measures such as visibility and sky exposure. While ArcGIS also supports the computation of shadows, it does not have the ability to accumulate shadows and visualize this accumulation.  

\myparagraph{Shadows in urban design.}
Sunlight exposure has been a core consideration in building design since early architectural studies. 
The seminal work of architect Ralph~L.~Knowles in which he proposed the concept of a \emph{solar envelope}~\cite{knowles1974energy} has been hugely influential on studies involving the impact of shadows, and more generally solar access. 
This was further explored in the following decades, by Knowles~\cite{knowles1982sun,knowles2003solar}, and others \cite{littlefair1998passive,littlefair2001daylight,capeluto2001use,compagnon2004solar,kampf2010optimisation}, stressing the importance of solar access in the urban context.
%

In the specific context of shadows, Richens and Ratti~\cite{richens1997image,ratti1999urban,ratti2004raster} proposed a technique that computes shadow information based on digital elevation models (DEM) and used it as a parameter to generate and evaluate urban models. 
Shadows computed using DEM have also been incorporated into urban climate models~\cite{lindberg2005towards,lindberg2008solweig,lindberg2011influence,lindberg2015solar}. However, in all these cases, the computed shadow information is approximate since DEMs are not only limited by the resolution of the images, they also do not capture the actual shape of the buildings.
Solar potential analysis also makes use of shadow information to 
improve the modeling of solar radiation, as well as assess photovoltaic energy-potential
of urban environments~\cite{FREITAS2015915}. 
These approaches also explicitly identify shadows for each time step of interest either 
through shadow maps~\cite{Catita:2014:ESP:2745545.2745576} or using ray casting~\cite{envirvis.20141103,JAKUBIEC2013127}.
Having an efficient approach to compute and accumulate shadows will greatly help in improving these models.

\myparagraph{Shadow computation techniques.}
The computation of shadow information has been extensively explored in computer graphics. We refer the reader to two recently published books by Eisemann~et~al.~\cite{Eisemann:2011:RS:2049989} and Woo~et~al.~\cite{woo2012shadow} for a detailed survey on recent shadow computation techniques. 
Real-time shadow computation can be broadly divided into two categories -- shadow map based techniques and shadow volume based techniques. The first group encodes shadow information of the scene geometry onto an image, which is later used when rendering the scene. Because the shadow information is discretized as an image, shadow map based approaches are constrained by the image resolution. 
Several solutions have been proposed to overcome this problem that include using multiple shadow maps~\cite{Zhang:2006:PSM:1128923.1128975,lloyd2006warping}, pre-computing and storing high-resolution maps~\cite{sintorn:2014:CPV:2601097.2601221,Kampe:2015:FMC:2699276.2699284}, and deforming the light projection matrix in order to increase the texel density near the view camera~\cite{Stamminger:2002:PSM:566654.566616,wimmer2004light}.
Sintorn~et~al.~\cite{Sintorn:2008} proposed a shadow mapping technique that maintains a list of points (from the camera's POV) corresponding to each pixel of the shadow map to avoid the aliasing artifacts. 
Lokovic and Veach~\cite{DeepShadowMaps} stored multiple depth values as a parametric function to compute shadows for dense translucent objects such as hair and fur. The proposed shadow accrual map also uses an approach of storing multiple depth values, but the method for computing this is different since it has to consider different time steps.
Scherzer~et~al.~\cite{scherzer2011survey} presented a detailed survey on shadow map based techniques.

Shadow volumes based techniques, on the other hand, do not perform any discretization of the scene. Instead, it uses the geometry of the scene to create \emph{volumes} of shadows in space. We refer the reader to Kolivand~et~al.~\cite{kolivand2013survey} for a survey on shadow volume approaches. 
Recently, Sintorn~et~al.~\cite{Sintorn:2014:PSV:2556700.2556716} proposed a shadow volume technique that assigns a volume for each triangle in the scene. 
However, since this requires pre-computation, impact computation in real-time becomes difficult. 
As mentioned in Section~1, accumulating shadows using any of the above techniques requires explicitly computing shadows at each time step, making it an expensive process.

With the current progress of massively parallel architectures, another approach that has become popular for computing shadows is ray tracing. Djeu~et~al.~\cite{djeu2009accelerating} used volumetric occluders to accelerate the tracing of rays for shadows. Kalojanov~et~al.~\cite{kalojanov2011two} stored the scene using a two-level grid to enable interactive ray tracing using GPUs. Feltman~et~al.~\cite{feltman2012srdh} proposed a cost estimator for shadow ray traversal, and used it to indicate early ray termination. Nah~et~al.~\cite{nah2014sato} proposed a surface method traversal order that accelerates shadow ray tracing. 

Soft shadowing techniques~\cite{Sintorn:2008,CGF:CGF1076,Fernando:2004:GGP:983868,Heidrich2000} can be used to obtain the desired visual effect of shadow accumulation, by considering samples to correspond to the time steps. This still boils down to explicitly computing shadows at each of the time steps. Also, using all of these techniques, it is not possible to quantify the shadow contribution with respect to the source, which is important for analysis. 

To the best of our knowledge, the only approach that computes shadows over time of day was proposed by Fernando~\cite[Chapter~13]{Fernando:2004:GGP:983868}, which pre-computes \textit{occlusion interval maps} that store for each point the time steps when they are visible to light. This is a costly pre-computation, which cannot be easily adjusted to interactively compute impact with changes to the scene.
Orthogonal to these techniques, our approach uses the property of shadow movement over time to accumulate shadows in real-time. Moreover, our approach can be used to extend any of the above techniques to speedup shadow accumulation.

\vspace{-0.2cm}
\section{Temporal Shadows}
\label{sec:temp-shadows}

In this section we first formally define two measures to quantify shadow accumulation followed by describing the key property of temporal shadows, that form the basis of our shadow accumulation techniques. 

\vspace{-0.2cm}
\subsection{Shadow Accumulation}
\label{sec:definitions}

A given location can be in shadow at different times of a given day. Our goal is to measure the quantity of shadow at these locations. In particular, we are interested in the following quantities which define two different aspects of a shadow with respect to a location:

\begin{figure}[t]
\centering
\includegraphics[width=0.55\linewidth]{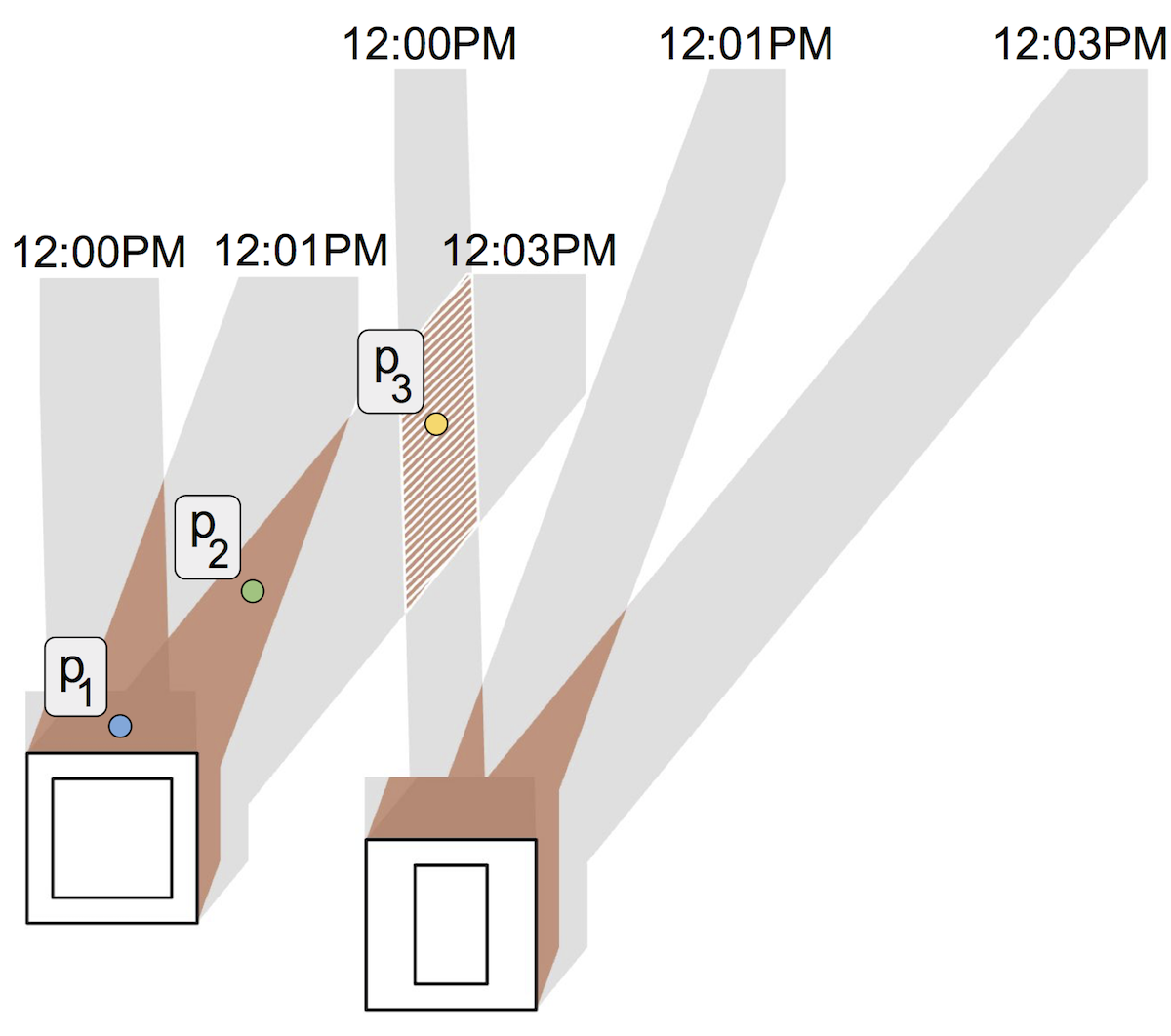}
\vspace{-0.2in}
\caption{Shadow accumulation is measured as either gross shadow or continuous shadow. In the above example, the point $p_2$ has both measures equal to 2 minutes in the 3-minute interval. The point $p_3$, on the other hand, is continuously in shadow for only a minute, even though it is in shadow for 2 minutes in this interval.}
\label{fig:shadowType}
\vspace{-0.2in}
\end{figure}

\myparagraph{Gross Shadow:}
measures the total time that a given location is in shadow during a given time interval. When computed over multiple days, we compute the gross shadow as the average time per day that location is in shadow.
For example, consider the shadows with respect to two towers in Figure~\ref{fig:shadowType} for a 3-minute interval\footnote{For illustrative purposes, the shadow between consecutive minutes are exaggerated. In reality, the shadows are much closer.}. 
Point $p_1$ is in shadow for the entire time interval, while points $p_2$ and $p_3$ are in shadow for 2 minutes of the interval.

\myparagraph{Continuous Shadow:}
measures the maximum time that a given location is continuously in shadow during a given time interval. When computed over multiple days, it is equal to the maximum continuous duration over all days. 
Again, consider the example in \reffig{fig:shadowType}. Points $p_1$ and $p_2$ are continuously in shadow for 3 and 2~minutes respectively, which is the same as the total time (gross shadow) they are in shadow. On the other hand, even though $p_3$ is in shadow for 2~minutes, by having no shadow at 12:01~PM, it is continuously in shadow only over 1-minute intervals.

\subsection{Properties of Temporal Shadows}
\label{sec:prop}

One way to accumulate shadows in a given time interval is to compute shadows at each time step of this interval and combine them. However, as mentioned earlier, this is a costly operation and is not interactively feasible even when the accumulation is done only over a single week. 
Rather than tracking the movement of the sun (directional light source) over time, the key idea behind our technique is to alternatively track the movement of the shadow itself in order to accumulate them. 
For the remainder of the paper we assume the light source as directional.

\begin{table}[b]
\vspace{-0.2in}
\scriptsize
\centering
\caption{Mean~($\mu$), standard deviation~($\sigma$) , and median of the cosine value between the actual direction of sun light (in NYC) and the direction obtained by linearly interpolating between different time interval sizes. Note that the linear approximation starts diverging from the actual direction only when the interval size is greater than an hour.}
\label{tab:linear-approx}
\vspace{-0.1in}
\begin{tabular}{|c|c|c|c|c|c|c|} \hline
 Minutes   & 5 & 10 & 30 & 60 & 120 & 240  \\ \hline \hline
 $\mu$ ($\times 10^{-1}$)      & 9.9999 & 9.9999 & 9.9999 & 9.9999  & 9.9997 & 9.9952  \\ \hline
 $\sigma$ ($\times 10^{-6}$)& 1.31 & 0.93 & 1.90  & 4.92  & 25.5 & 415  \\ \hline
 Median    & 1  &  1  &  1  &  0.99999   &   0.99998   &   0.99965 \\ \hline
\end{tabular}
\end{table}

Consider the time interval $[t_1, t_n]$. Given a relatively short time interval, the movement of the sun during this interval can be considered to be linear. 
To validate this assumption in practice and to identify an appropriate time interval, we compare the actual sun direction with the interpolated direction over different interval sizes. 
In particular, we first choose $1000$ random time steps covering the entire year. This ensures that directions from different times of the day as well as different seasons are well covered.
Given a time interval of $n$ minutes, we compute the cosine between the actual sun direction at each minute and the direction obtained by interpolating between the directions at the start and end of that interval. A value close to 1 indicates that the two direction vectors are the same.
Table~\ref{tab:linear-approx} shows the mean, standard deviation, and median of the cosine values with different interval sizes ranging from 5 minutes to 4 hours. Note that the linear assumption of the sun movement holds even when the value of $n = 60$~minutes. We start seeing the interpolated directions diverging from the actual direction beyond this interval. 
We therefore decided to use \textit{hourly intervals} ($n = 60$~minutes) to compute shadow accumulation.

\begin{figure}
\centering
\includegraphics[height=3.5cm]{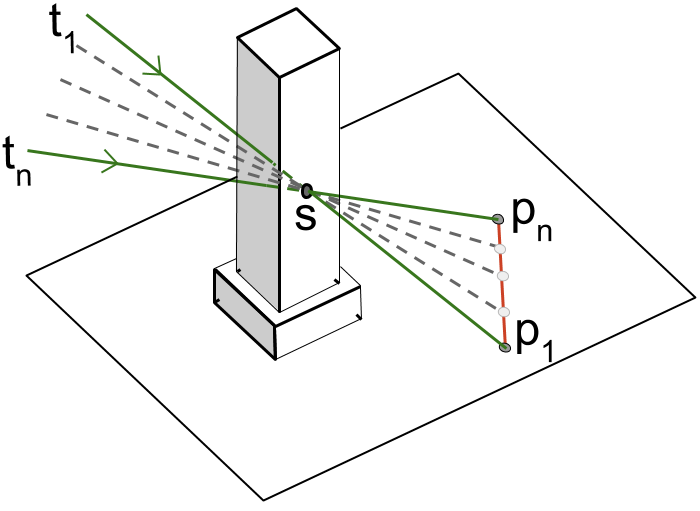}
\vspace{-0.1in}
\caption{Shadow accrual map makes use of the linear movement of the sun over short time periods to track the movement of shadows. Given a time interval, $[t_1,t_n]$, each point $p_1$ in shadow at $t_1$ is mapped to the point $p_n$, the location of shadow at $t_n$ due to the same shadow source, $s$.} 
\label{fig:sam}
\vspace{-0.25in}
\end{figure}

The main idea behind our approach is the following.
Consider point $p_1$ on the ground that is in shadow at time $t_1$. As illustrated in \reffig{fig:sam}, let the cause of shadow at $p_1$ be the point $s$ on a building.
Note that $s$ can be one of the many possible sources of shadow at $p_1$.
Let at time $t_n$, $s$ cast a shadow at point $p_n$. Then, given that the movement of the sun is linear, the shadow cast by $s$ moves linearly from $p_1$ to $p_n$. 
Thus, the accumulated shadow corresponding to $s$ over the given time interval is essentially the straight line from $p_1$ to $p_n$. This key observation is used in the next two sections to design algorithms to efficiently compute shadow accumulation over time.

\section{Shadow Accrual Maps}
\label{sec:sam}

We now describe \textit{shadow accrual maps}, an extension to the standard shadow mapping technique~\cite{shadowmap} that utilize the linear shadow observation to keep track of shadows over time. 
The standard shadow mapping algorithm runs in two passes. 
First, it renders the scene from the point of view of the light (using an orthographic projection in the case of a directional light), 
and stores the depth buffer in a texture called the shadow map. 
The shadow map maintains the distances between the light and the objects that are directly illuminated. In case of directional light, the distance is computed with respect to a plane orthogonal to the light direction.
Next, the scene is rendered from the point of view of the camera. The depth of the surface point corresponding to each pixel is computed from the light's point of view as above. If this depth is greater than the depth stored in the corresponding pixel in the shadow map, the point is marked as being in a shadow. 

Our algorithm follows the same template, wherein the first step computes the shadow accrual map for a given fixed time range, and the second step identifies points in shadow.

\myparagraph{Step 1: Computing shadow accrual maps.}
Consider a given time range $[t_1, t_n)$ in which the movement of the sun is linear. 
Let this time range be divided into $n$ discrete time steps. The shadow accrual map is a 3D texture that stores the depth values corresponding to these $n$ time steps. However, instead of individually computing the $n$ 2D textures (or shadow maps) over $n$ passes, we compute it in one pass as follows.

\begin{figure}[t]
\centering
\includegraphics[height=3.5cm]{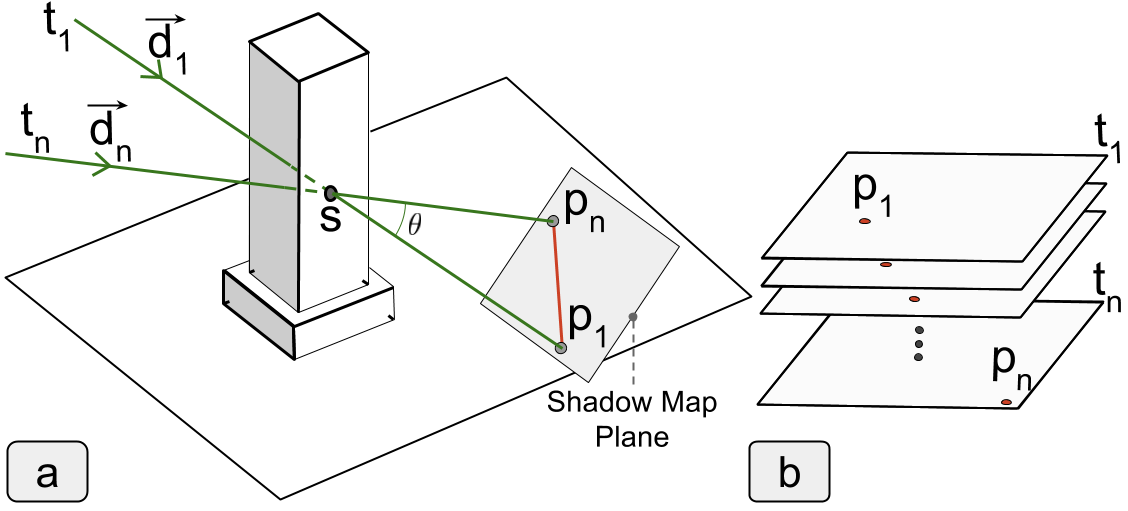}
\vspace{-0.1in}
\caption{Shadow accrual map is a 3D texture, where each slice stores the depth values corresponding to a single time step. The depth value for a given time is assigned by interpolating the shadow from its projection at time $t_1$ to its projection at time $t_2$.}
\label{fig:shadow-map}
\vspace{-0.25in}
\end{figure}

Let $\vec{d_1}$ and $\vec{d_n}$ be the direction of sun light at 
the beginning and end of the given time range. 
We select a shadow plane that is orthogonal to $\vec{d_1}$, and further 
from the light than all objects visible from the camera as shown in \reffig{fig:shadow-map}a. 
The extent of this plane is computed such that it encompasses all
objects from the point of view of the camera when projected with respect 
to directions $\vec{d_1}$ and $\vec{d_n}$.
Every 3D point $s$ in the scene is processed as follows. 
$s$ is first projected onto the shadow plane along directions $\vec{d_1}$ and $\vec{d_n}$ 
to obtain points $p_1$ and $p_n$ respectively.
That is, $p_1$ and $p_n$ correspond to the locations of the shadow cast 
by $s$ at times $t_1$ and $t_n$ respectively (see \reffig{fig:shadow-map}a).
Since the shadow moves linearly within the given time interval, the location $p_i$ of the shadow at
every intermediate time step $t_i$, $1 < i < n$, is computed as
$p_i = p_1 + (p_n - p_1) \times \frac{\tan(i'\theta / n)}{\tan(\theta)}$,
where $\theta = \angle(\vec{d_1},\vec{d_n)}$, and $i' = i-1$. 
For each $i$, shadow depth of $p_i$ in the $i^{\textrm{th}}$ 2D slice is appropriately updated as shown in \reffig{fig:shadow-map}b.
Instead of the distance between $s$ and the light, we use an 
equivalent measure of the distance of $s$ to the shadow plane along the light direction as depth.
Thus, the depth of $p_i$ is simply the distance between $s$ and $p_i$.

Using modern programmable GPUs, the entire operation can be 
performed in parallel in a single rendering pass. 
As the following theorem shows, the resulting 3D texture is equivalent to creating independent shadow 
maps for each of the $n$ time slices.  
Thus, there is no loss of quality when using shadow accrual maps compared to traditional shadow maps.

\begin{theorem}
Consider a time interval of size $n$ units during which the movement of the sun is linear. Let shadows be accumulated for every 1 unit of time, \ie the time interval is divided into $n$ equal time steps. Then the shadow accrual map computed for this time interval is the same as the computing $n$ shadow maps for each of the $n$ time steps.
\end{theorem}

\begin{proof}
Consider a time interal $[t_1,t_n]$ of length $n$ units. Without loss of generality, let $t_1 = 0$ and $t_n = n-1$. Let $\vec{d_1}$ and $\vec{d_n}$ be the direction of light at $t_1$ and $t_n$ respectively. By definition, the first and last slices (corresponding to time $t_1$ and $t_n$) of the shadow accrual map are the same as the shadow maps for these two time steps.

 {\makeatletter
\let\par\@@par
\par\parshape0
\everypar{}
\begin{wrapfigure}{r}{0.5\linewidth}
\centering
\vspace{-0.1in}
\centering
\includegraphics[width=0.9\linewidth]{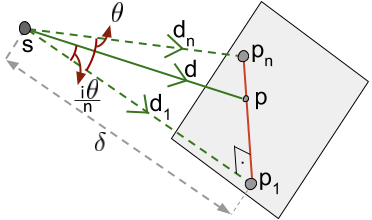}
\vspace{-0.1in}
\caption{The pixel $p$ processed by the shadow map at time $i$ is obtained by projecting along the light direction $\vec{d}$, which is at an angle $i\theta$ from $\vec{d_1}$.}
\label{fig:proof}
\vspace{-0.2in}
\end{wrapfigure}
Now, consider time step $t_1 < i < t_n$, having direction $\vec{d}$. Due to linear movement of the sun, the interpolation factor $k = \frac{i - t_1}{t_n - t_1} = \frac{\angle(\vec{d},\vec{d_1})}{\angle(\vec{d_n},\vec{d_1})}$.
Let $\angle(\vec{d_n},\vec{d_1}) = \theta$. Then $\angle(\vec{d},\vec{d_1}) = i\theta/n$.
Consider a point $s$ on building mesh. Let the projection of $s$ on the shadow plane be on point $p$ at time $i$. Without loss of generality, we use $p$ to denote the pixel on the shadow map texture as well. Therefore, the point $p$ will be processed by the traditional shadow map algorithm for time step $i$. 
Let the shadows at times $t_1$ and $t_n$ be at points $p_1$ and $p_n$ respectively.
Note that the direction $\vec{d_1}$ is orthogonal to the shadow plane. Therefore, the point $p$ is at a distance $\delta \tan(i\theta/n)$ from $p_1$ along the line $[p_1,p_n]$  (see \reffig{fig:proof}).
\par}

Now consider the shadow accrual map algorithm. At time $i$, it processes the point $p' = p_1 + (p_n - p_1) \times \frac{\tan(i\theta/n)}{\tan(\theta)}$, which is the same as $p$, for the $i^{th}$ slice.
Thus, for any shadow source $s$, shadow accrual map processes the exact same pixels for all time steps $i$ as a shadow map would for the corresponding directions.
\end{proof}

\myparagraph{Step 2: Computing shadows.}
To obtain the shadow at a given 3D point $s$, we need to test its depth at the $n$ time steps. 
To do this, consider again the line between the projection of $s$ at time $t_1$ and time $t_n$. 
As before, let $p_i$ be the projection of $s$ at time steps $1 < i < n$. 
If $s$ is in shadow at time step $i$, then the depth of $p_i$ will be less than 
the corresponding depth stored in the $i^{\textrm{th}}$ 2D slice of the shadow accrual map 
(recall that the depth is the distance between $s$ and $p_i$).
The gross shadow is then computed by simply counting the number of points $p_i$, $1 \leq i \leq n$, that are in shadow.
The continuous shadow is computed by counting the maximum number of points that are continuously in shadow.
Note that if each time step is different from 1 minute, then the gross (continuous) shadow is multiplied by an appropriate factor.

\section{Inverse Accrual Maps}
\label{sec:iam}

The primary application of accumulated shadows in the context of cities is in studying its impact on open urban spaces (such as parks or sidewalks), which are typically flat surfaces.
Since shadow accrual maps are based on shadow maps, the well known issues 
such as aliasing and shadow acne are also carried over making accurate quantification of shadows difficult. 
Ray tracing based shadow techniques, on the other hand, have better quality. 
With the focus on shadow accumulation over flat surfaces, in this
section, we design a ray tracing-based approach that makes use of
the linear movement property of temporal shadows.

\myparagraph{Inverse Accrual Maps.}
Consider again the example in \reffig{fig:sam}. Given that the accumulated shadow corresponding to $s$ is the straight line from $p_1$ to $p_n$, the inverse accrual map maps the point $p_1$ to the point $p_n$. It is computed as follows.
First, the set of points on the plane visible from the camera are identified. This can be accomplished by a simple modification of the rendering output of the graphics pipeline where the world coordinates of each pixel is stored onto a buffer.

\begin{wrapfigure}{r}{0.5\linewidth}
\centering
\vspace{-0.1in}
\includegraphics[width=\linewidth]{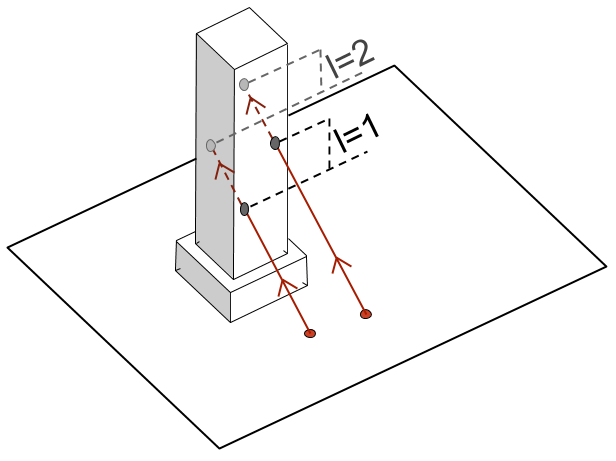}
\vspace{-0.3in}
\caption{A given view point can have more than one source of shadow. The $l^{\textrm{th}}$ closest source of shadow is used to represent the $l^{\textrm{th}}$ 2D slice of the inverse accrual map.}
\label{fig:ism}
\vspace{-0.2in}
\end{wrapfigure}
Now, consider any point, say $p_1$, on a plane. The possible sources of shadow for that point can be obtained by tracing a ray from that point in the reverse direction of light until there are no more intersections. Here, each intersection corresponds to a source of shadow. The inverse accrual map is also a 3D texture, where the $i^{\textrm{th}}$ 2D slice stores the mapping corresponding to the $i^{\textrm{th}}$ source. Figure \ref{fig:ism} shows two points for which there are two source points.
If there is no shadow on $p_1$ at $t_1$, then such a point has no source of shadow, and is hence mapped to infinity.
To avoid shadow acne, we ensure that there is a small distance between $p_1$ and its shadow source (we notice that in practice, a distance of 100~cm provides good results).

\myparagraph{Computing shadow accumulation.}
Let the given time interval be divided into $n$ time steps. The shadows corresponding to each of the source levels are first drawn as follows. Consider a point $p_1$ and the corresponding mapping point $p_n$. The shadow at each time step is approximated to be along one of the line segments obtained by dividing the line $(p_1,p_n)$ into $n-1$ segments. We maintain a $n$-bit vector for each point to store the shadow corresponding to it at the different time steps.
Consider a point $p$. The bit corresponding to the $j^{\textrm{th}}$ time step is set to 1, if the $j^{\textrm{th}}$ line segment pass through this point. 
For example, consider the illustrated points in \reffig{fig:shadowType}. Given the 3-minute time interval with $n = 3$, each point has a 3-bit vector associated with it. The vector corresponding to $p_1$, $p_2$, and $p_3$ are $[1,1,1]$, $[0,1,1]$, and $[1,0,1]$ respectively. 

After all $n$-bit vectors corresponding to points in the scene are populated by drawing all valid lines from all source levels, gross shadow is computed as the sum of bits in this vector. Continuous shadow is computed as the maximum size of a set of consecutive 1's in the vector. As before, gross (continuous) shadow is multiplied by an appropriate factor to offset the size of a time step.

\begin{wrapfigure}{r}{0.4\linewidth}
\centering
\vspace{-0.4cm}
\includegraphics[width=\linewidth]{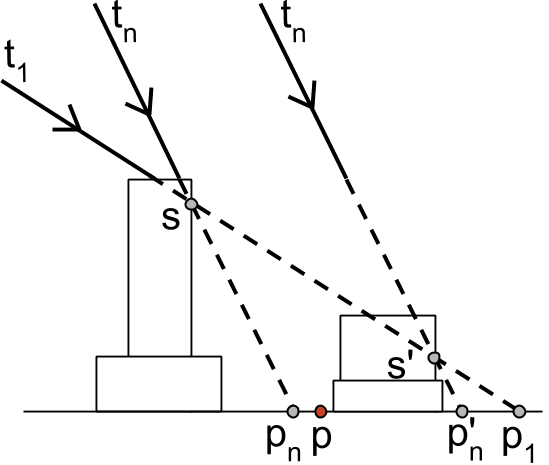}
\vspace{-0.2in}
\caption{Possible situation which requires inverse accrual maps to be computed at multiple source levels.}
\label{fig:levels}
\vspace{-0.15in}
\end{wrapfigure}

\myparagraph{Effect of maximum source level.}
\label{sec:trade-off}
Consider the evolution of the shadow at point $p$ in \reffig{fig:levels} from time $t_1$ to time $t_n$. Let the line $(p_1, p_n)$ be responsible for the shadow at $p$ at time $t$, where $t_1 < t < t_n$. Let $s$ be the source of this shadow.
For the point $p$ to be correctly identified as being in shadow at $t$ for source level $l = 1$, the corresponding inverse accrual map should associate point $p_1$ to $p_n$. However, as shown in \reffig{fig:levels}, this is not true because source level $l = 1$ uses the closest shadow source, which in this case is not~$s$. 

In order to obtain accurate shadow accumulation, it is therefore necessary to compute inverse accrual maps over all possible source levels. 
This becomes expensive especially during dawn or dusk, since the light direction from the sun is close to horizontal and a ray in the reverse light direction can intersect several buildings. However, in such a scenario, the movement of shadows due to farther away points is very fast, causing little loss of accuracy if these points are omitted.
Thus, to maintain practical computation times, we can limit the maximum value of $l$ when computing the inverse accrual map. 
Also, when accumulating shadows from time $t_1$ to $t_n$, in addition to computing inverse accrual maps from $t_1$ to $t_n$, we also compute them from $t_n$ to $t_1$. This serves three functions: 
(1)~during the later part of the day, the shadow stretches with time. Thus, a point in shadow at $t_1$ would correspond to an area of points at $t_n$. However, the inverse accrual map associates only a single point at $t_n$, and the drawn shadow line will not reflect this stretch. By reversing the time interval, the map would then correspond to contracting shadows and will ensure that no points are missed during shadow accumulation; 
(2)~in the example from \reffig{fig:levels}, note that $p_n$ for the time interval $(t_n,t_1)$ is mapped to $p_1$. Thus, the point $p$ is not problematic anymore when $l=1$. This also helps improve accuracy while still maintaining a small number of 2D slices; and
(3)~when the ray at time $t_1$ is parallel to a building facade, these instances are captured in inverse accrual maps from $t_n$ to $t_1$, thus improving the accuracy of the approach.   

For the remaining of the paper, when using inverse accrual maps, we compute the map along both $t_1$ to $t_n$ and $t_n$ to $t_1$. With this addition, as we show later in Section~\ref{sec:exp}, the accumulated shadow converges close to its true value with very low error when the source level $l \leq 3$.

\myparagraph{Discussion.}
Given that the computation of inverse accrual maps identifies the sources of shadow, a simple modification to keep track of this will allow the identification of the source of the shadow -- the object(s) causing the shadow. As we show later, this is useful for analyzing shadows and their causes in cities.
%

\section{Handling Large Time Intervals}
\label{sec:large-time}

The shadow accumulation using either of the above two approaches is computed for short time intervals (60 minutes) when the movement of the sun can be approximated to be linear. Therefore, when required to accumulate shadows spanning multiple days (or months), one way to accomplish this is to explicitly compute shadow accrual maps for all 60-minute intervals at a resolution of 1 minute (\ie $n = 60$) corresponding to the given time period. 

While the direction of sun light at a given time in summer will be drastically different from the direction in winter at the same time (depending on the geographical location), the change in direction on consecutive days in summer (or winter) is minimal.
We use this key observation to significantly reduce the number of shadow accrual maps (or inverse accrual maps) that are computed, as shown next.

For a city of interest, in a preprocessing step, we first cluster all possible light directions into a set of bins. Consider a ray along each light direction originating from a reference point, which is the origin. Then, the bins are defined by partitioning a hemisphere, that is centered at this origin, into quads such that the maximum angle (azimuthal and polar angle) corresponding to any quad is bounded.
Using a sufficiently small bound, any light direction can be represented by the bin it is associated with. 

Now, let the shadow be accumulated from time $t_{start}$ to $t_{end}$ for a period of $d$ days. So each day will require shadow accrual maps (or inverse accrual maps) to be computed for $k = (t_{end} - t_{start}) / n$ time intervals per day. This can be represented as a collection of pairs $(\vec{s_j},\vec{e_j})$, $1 \leq j \leq k$, where $\vec{s_j}$ is the start light direction and $\vec{e_j}$ is the end light direction for the $j^\textrm{th}$ 60-minute interval.
We create a weighted graph $G(V,E)$, called \textit{direction graph}, where each node in $V$ corresponds to one bin of the above described index. There is an edge between two nodes if there exists a direction pair $(\vec{s_j},\vec{e_j})$ corresponding to those bins. The weight of an edge is the number of times that pair is present for the given time interval.
\begin{wrapfigure}{r}{0.55\linewidth}
\centering
\includegraphics[width=\linewidth]{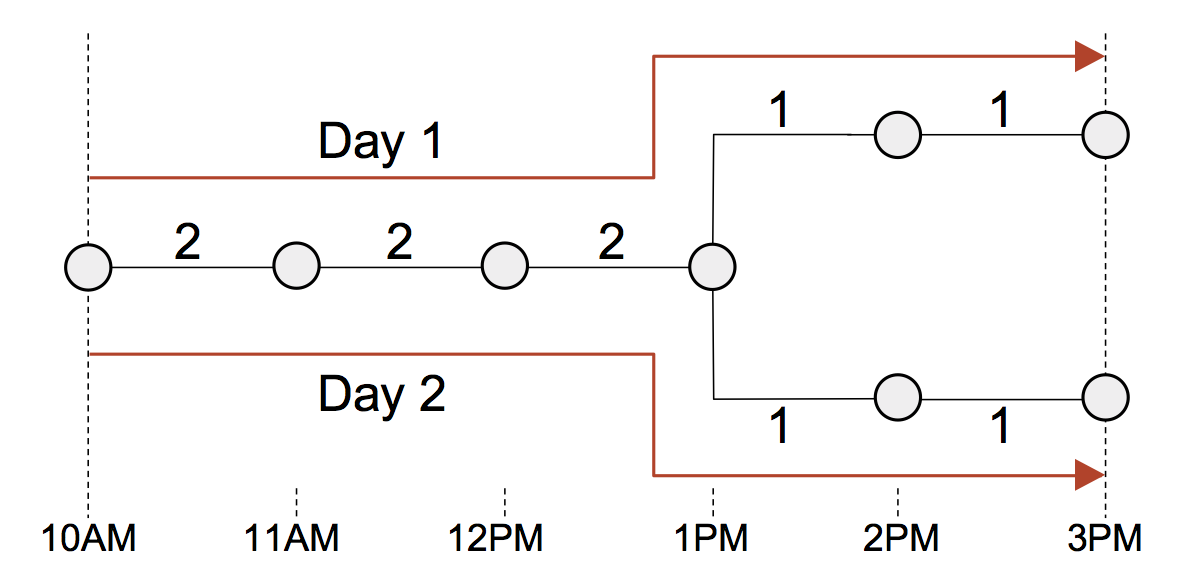}
\vspace{-0.2in}
\caption{The direction graph is used to significantly improve the performance of shadow accumulation.}
\label{fig:graph}
\vspace{-0.2in}
\end{wrapfigure}
For example, consider a case where shadows have to be accumulated over two days from 10~AM to 3~PM. Let the value of $n = 60$~minutes. If the direction of sun light remains the same for both days until 1~PM, then the resulting graph is as shown in \reffig{fig:graph}. For the first 3 hourly intervals, the edge weights will be 2 since the corresponding directions are common between the two days.

Accrual maps now have to be computed only once for each edge. This number is significantly smaller than explicitly computing them for all $k \times d$ intervals (see Section~\ref{sec:exp} for more details).
Thus, in the above example, shadow accrual maps (or inverse accrual maps) have to be computed only for 7 one hour intervals instead of 10 one hour intervals (\ie $k = 5$ hourly intervals over $d = 2$ days).

Given this setup, the different shadow accumulation quantities are computed as follows.

\myparagraph{Gross shadows.}
A given $n$-minute interval corresponds to an edge in the direction graph.
Let $G_j$ be the gross shadow computed for edge $j$ in the graph.
When considering multiple such intervals, the gross shadow is equal to the sum of gross shadows computed from each interval. 
Given the direction graph, this sum is equal to
\decsp
\begin{equation}
    G = \sum_{j = 1}^{k} G_j \times w_j
\decsp    
\end{equation}
where $w_j$ is the weight of the corresponding edge.

\myparagraph{Continuous shadows.}
Consider an edge in the direction graph and the associated shadow accrual maps. When computing continuous shadows for each point in the given interval, in addition to the maximum continuous shadows for the corresponding interval, we also store the length of the longest prefix and longest suffix of continuous shadows (these will be the longest prefix and suffix of 1's from the $n$-bit vector in case of using inverse accrual maps).
The movement of sun on each day corresponds to a path of edges in the direction graph. For the example in \reffig{fig:graph}, paths corresponding to the two days is illustrated in red.
To compute the continuous shadow over these edges, the corresponding accrual maps are processed in the order of the path traversed. In particular, the prefix, maximum and suffix values are used to ``stitch" together consecutive edges.
Note that these values are computed only once for each edge, and reused multiple times. 
To avoid the number of accrual maps that are cached in memory, we traverse the paths in a topological order so that accrual maps can be discarded as soon as all paths using them are processed.
\begin{figure*}[t]
\centering
\includegraphics[width=0.9\linewidth]{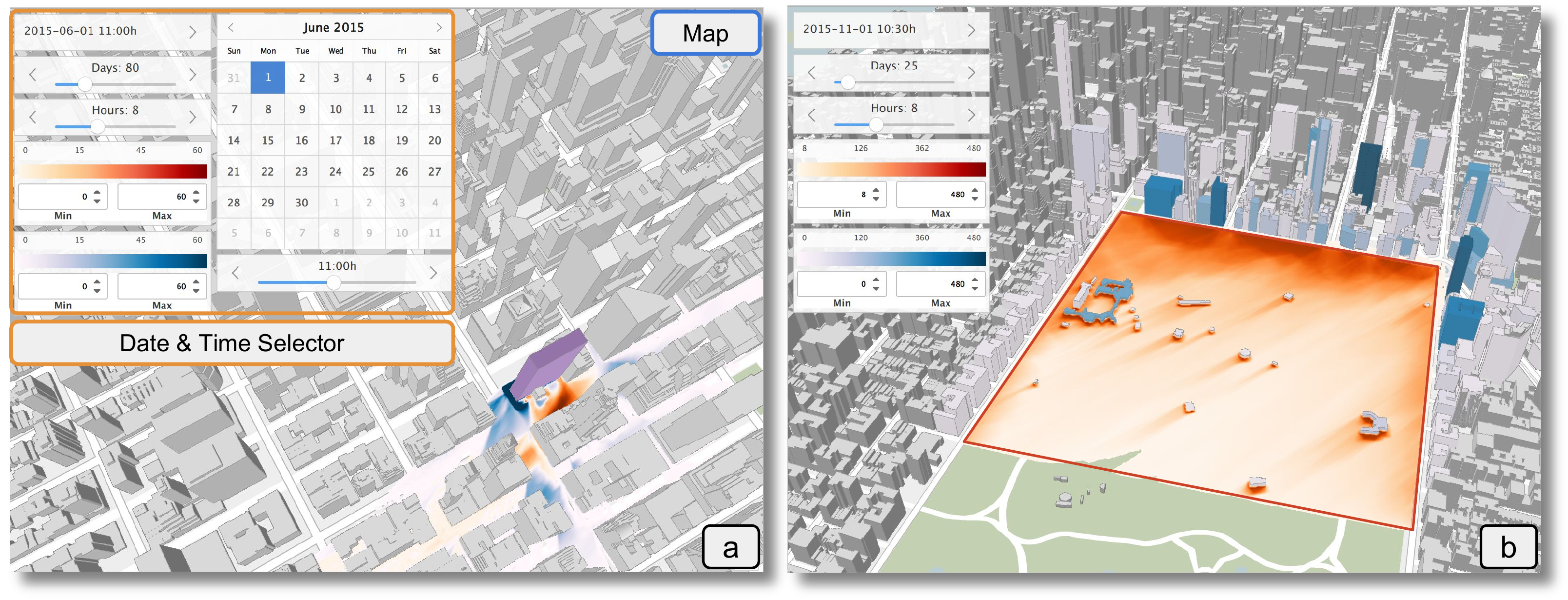}
\vspace{-0.2in}
\caption{User interface of the Shadow Profiler system consists of a map widget together with a date and time selector.
\textbf{(a)}~We analyze the shadow impact when inserting a new building. The divergent color map highlights areas where the new building would add shadows (in red), or areas where it would decrease shadows (in blue). The scale of the color map is in minutes, and can be adjusted by the user. 
\textbf{(b)}~Visualizing the shadow contribution of buildings with respect to accumulated shadows in the selected region; a darker shade of blue indicates  a higher contribution by that building. The accumulated shadows are visualized using the color map shown in the interface.}
\label{fig:gui}
\vspace{-0.2in}
\end{figure*}

\section{Shadow Profiler}
\label{sec:framework}
The shadow accumulation approaches are used to design \emph{Shadow Profiler}, a visual exploration system that allows users to explore and analyze shadows in a city. 
%
%
We now briefly describe its visual interface and discuss the analysis measures it supports.
The accompanying video demonstrates the different components of Shadow Profiler.

\subsection{Visualization Interface}
\label{sec:vis}

The interface (\reffig{fig:gui}) is primarily composed of two components: 
1.~a \textit{3D map} widget that provides spatial context; and 
2.~a \emph{date \& time selector} widget that allows for the user to select a time period of interest. A time period is selected by specifying four values -- a start date $D_{start}$, start time $t_{start}$, number of days $n \geq 1$, and hours per day $k \geq 0$.
A value of $k = 0$ represents a single time instant, and the shadow corresponding to the selected date and time is visualized.
When $k > 0$, shadows are accumulated for $k$ hours per day -- from $t_{start}$ to $t_{start} + k$ over a period $n$ days starting from $D_{start}$.
The accumulation type, which is one of gross or continuous shadow, is selected by the user. The accumulated shadow is averaged over the number of days, and is visualized using a color map (\reffig{fig:gui}b).
User interactions (pan, zoom, \etc) recompute shadows on the fly for the 
region corresponding to the viewport, thus enabling a level-of-detail rendering.
Figure \ref{fig:lod} visualizes the shadow accumulation over a single day at three different zoom levels focusing on 
Washington Square Park in New York City.
Accumulating shadows over a large number of days can still take few seconds depending on the approach used 
(see Section~\ref{sec:exp}). 
To support seamless interaction, we allow for a progressive computation and rendering of the shadow accumulation.
We also allow users to brush and select polygonal regions of interest to inspect shadows. In this case, the visualization is restricted to the specified polygon.

An important task is the assessment of shadow impact with respect to a new building. To support this, we allow the user to select either an empty building lot, or an existing building that is to be demolished, and replace it with a user generated mesh. The shadows are then updated to reflect this change; this is accomplished by computing the difference of shadows between the two states and visualizing the result using a divergent color map. Figure \ref{fig:gui}a illustrates this task.

We support two visualization modes corresponding to the two approaches, as shown in \reffig{fig:modes} and in  the accompanying video. Users can choose the mode based on their objective -- the \textit{exploration mode} is used for interactive visualization when users are interested in exploring the city, and uses shadow accrual maps; and the \textit{analysis mode} is used when users are interested in a more detailed analysis and for computing the different analysis measures (described in the next section), and uses inverse accrual maps.

\begin{figure*}[t]
\includegraphics[width=0.99\linewidth]{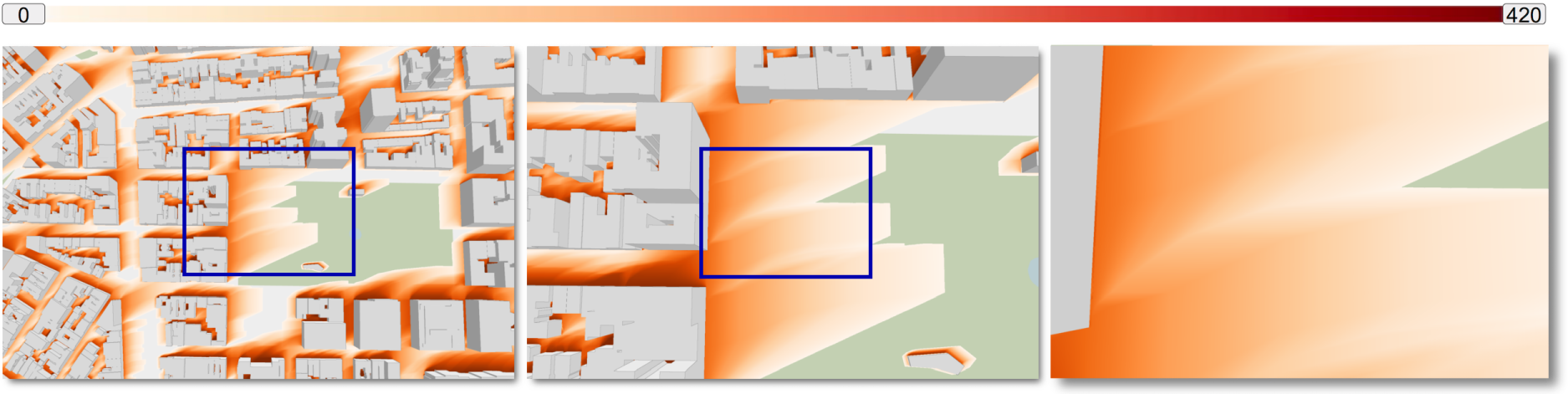}
\vspace{-0.2in}
\caption{Shadow accumulation over 7 hours on June 1 at three different zoom levels when the camera is 800, 300 and 100 meters, respectively, above the ground. Note that the level of detail, as well as the quality of the visualization improves as the user focuses into a region of interest.}
\label{fig:lod}

\vspace{0.1cm}
\includegraphics[width=0.99\linewidth]{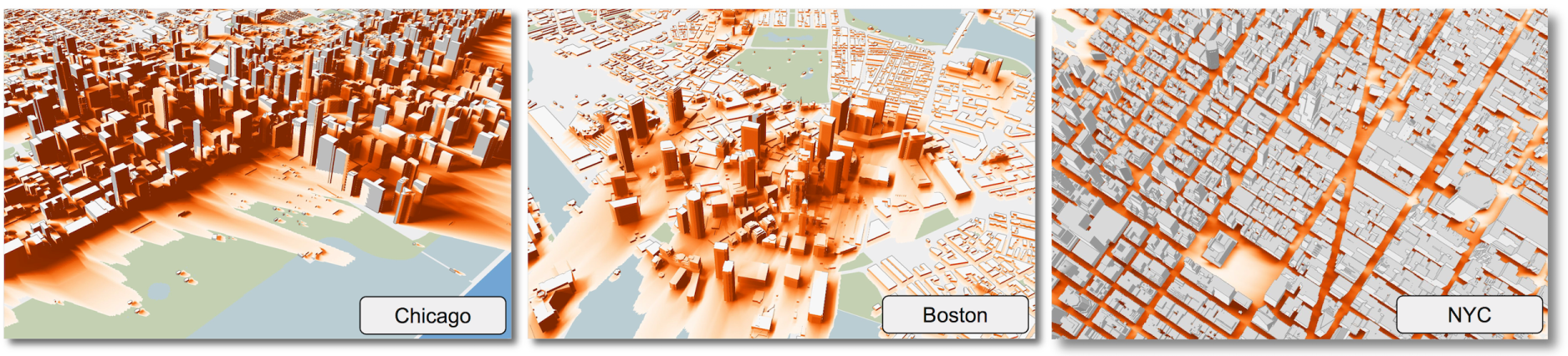}
\vspace{-0.2in}
\caption{Shadow profiler supports two modes of operation. The exploration mode, which uses shadow accrual maps to compute shadow accumulation, is used to explore Chicago and Boston. The analysis mode, which uses inverse accrual maps, is used for New York City. The shadows were accumulated for 7 hours on March~28.}
\label{fig:modes}
\vspace{-0.2in}
\end{figure*}

\subsection{Analysis Measures}
\label{sec:scores}

In addition to visualizing the accumulation, we also compute three different metrics quantifying the properties of the shadow. All of these quantities are computed with respect to a polygonal region $R$ of interest selected by the user. 

\myparagraph{Shadow area.}
Let $p \in R$ be a point that is within the selected region. The shadow area is computed as 
$Area = \int_{p \in R} shadow(p)$ 
where $shadow(p)$ is defined as follows. When a single time instant is being visualized, then $shadow(p) \in \{0,1\}$ indicating the absence / presence of a shadow. When accumulating, $shadow(p)$ indicates the fraction of the time (gross or continuous) per day that point is in shadow.
In a discrete setting, this value is equal to $\sum_{p \in R} shadow(p) \times area(p)$, where $p$ represents a pixel, and $area(p)$ is the area covered by the pixel in square meters. To maintain accuracy, shadows rendered at a high resolution are used for this computation.

\myparagraph{Shadow score.}
As mentioned earlier the effect of shadows can be both positive as well as negative. For example, from a pedestrian point of view, shadows are preferred during summer since it makes the environment more comfortable, while disliked in winter. 
To evaluate this effect, we define the shadow score as 
\decsp
\[Score = \int_{p \in R} \sum_{t \in T}(\omega_t \times shadow_t(p)) \decsp \] 
Here, the user divides the selected time period $T$ into a set of time intervals $t$, and assigns a weight -$1 \leq \omega_t \leq$ +1 for each interval, indicating the nature of shadow for that interval. For example, the user could assign a weight -1 for winter months, +1 for summer months, and 0 for other months. $shadow_t(p)$ specifies the fraction of the time per day a given point is in shadow during the interval $t$.
In addition to computing the score, it is also visualized using a divergent color map (see \reffig{fig:use22}). 

\myparagraph{Building contribution.}
The framework also allows for evaluating the shadow contribution of buildings over the selected region $r$. Here, each building is assigned a quantity equal to the shadow area resulting from that building, and visualized using a color map as shown in \reffig{fig:gui}b.

\section{Implementation and Experiments}
\label{sec:imp-exp}

The shadow accumulation techniques were implemented using C++, OpenGL~4.3, and OpenCL~1.2. We now briefly describe the implementation choices made, and then discuss results from our experiments evaluating the performance of the two approaches.

\begin{figure}[t]
\centering
\includegraphics[width=\linewidth]{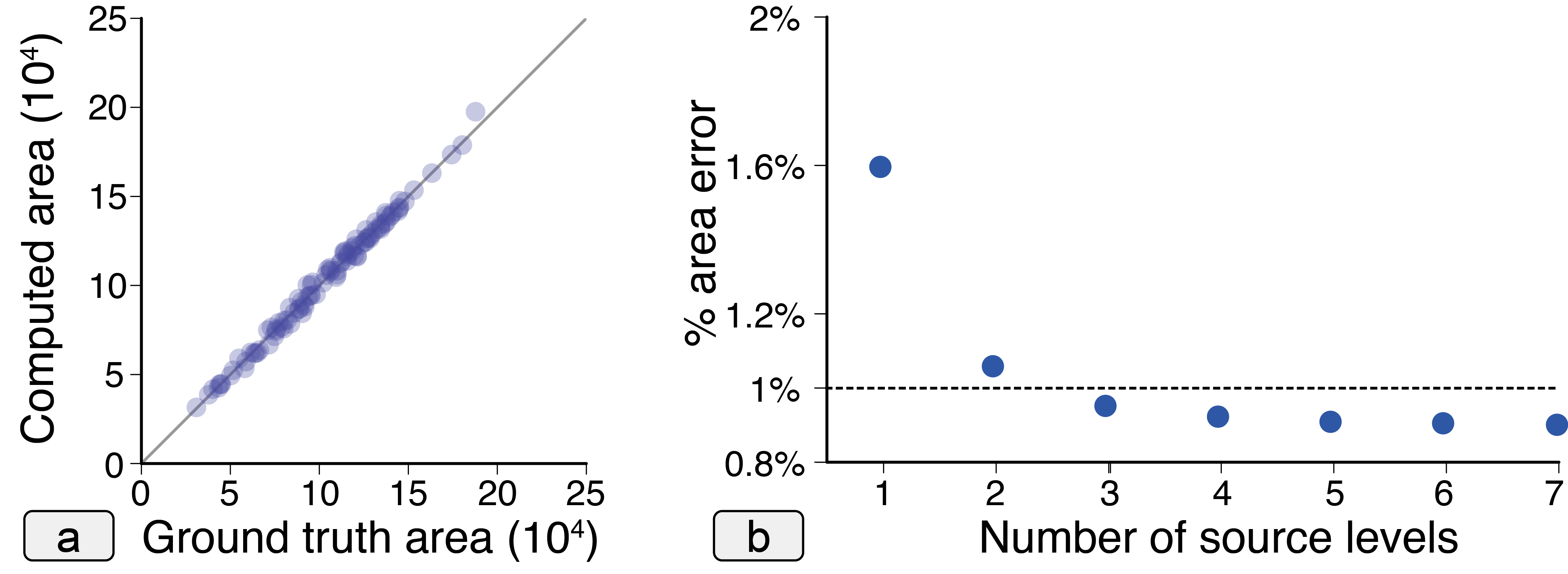}
\vspace{-0.3in}
\caption{
\textbf{(a)}~Comparing the area of accumulated shadows over 1 hour periods computed using a smaller set of representative directions from the
 direction graph with ground truth area. Note that approximating using the direction graph does not hamper the accuracy of the shadow area.
\textbf{(b)}~Choosing the maximum source level for inverse accrual maps. 
}
\label{fig:plots}
\vspace{-0.3in}
\end{figure}

\begin{figure*}[t]
\centering
\includegraphics[width=\linewidth]{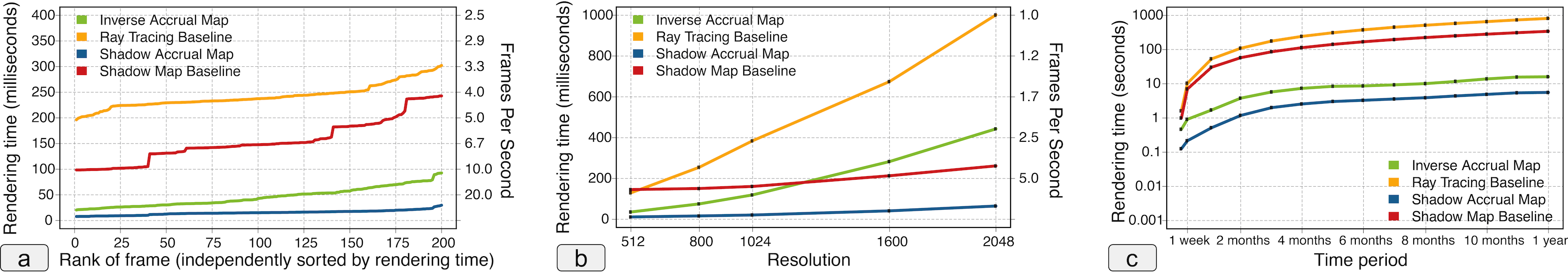}
\vspace{-0.3in}
\caption{Performance Evaluation.
\textbf{(a)}~Comparison with baselines. For each method, the times are independently sorted in increasing order. Note that both shadow accrual map and inverse accrual map consistently perform better than the naive baselines.  
\textbf{(b)}~Scalability of the proposed techniques with increasing resolution. We used resolutions with an aspect ratio of 1:1 for this experiment (\eg 512 implies a resolution of $512\times512)$. 
\textbf{(c)}~Scalability with increasing time periods. Note the significant speedup (over 50X) achieved with increasing time periods (y-axis is in log scale).
}
\label{fig:perf}
\vspace{-0.2in}
\end{figure*}

\subsection{Implementation}

\myparagraph{Shadow accrual maps.}
As mentioned earlier, the problems that are common with shadow maps also carry over to shadow accrual maps. 
To obtain better quality shadows at a lower resolution itself, we chose to use the trapezoidal transformation~\cite{Martin2004}, which warps the shadow depth texture onto a trapezoidal approximation of the view frustum. This however doesn't solve the problem of shadow acne. So we use a bias offset to reduce the effect of shadow acne.

When computing the shadow accrual maps, since OpenGL does not allow more than 8 frame buffer objects (and thus depth buffers), we use the features of OpenGL 4.3 to store the depth values onto a 3D image texture. By making use of atomic operations on images, we are able to store, in a single rendering pass, the largest depth value of a texel, for all slices of the shadow accrual map.

\myparagraph{Inverse accrual maps.}
As mentioned in Section~\ref{sec:iam}, computing inverse accrual maps is accomplished by tracing a ray along the reverse light direction. Given a maximum source level (see Section~\ref{sec:exp}), the ray is traced until either the given number of intersections is reached, or no other intersections are possible.
Our implementation uses a 3D grid to index the model of the city to be used for ray tracing. The corresponding accrual map associations are simultaneously computed during the ray traversal to output the inverse accrual maps. This part was implemented using OpenGL shaders. 
Note that the ray traversal also takes into account new buildings that are added or replaced.

OpenCL is then used to accumulate shadows -- \ie draw the shadow lines and perform the appropriate bit operations based on accumulation type. When rendering large time intervals, the computed values are combined with the existing values to enable progressive rendering. User specified operations such as computing analysis measures and impact are also performed at this stage.
%

\myparagraph{Baseline implementations.}
We implemented two baselines, based on shadow maps and ray tracing, to evaluate the performance of the proposed approaches. 
The shadow map baseline explicitly computes shadow maps for every minute, and uses them to identify and accumulate shadows. To maintain a consistent quality, we use the trapezoidal transformation for this implementation as well.

The brute force ray tracing based approach explicitly identifies the shadow for every minute in the time interval by tracing a ray from all pixels, which are then accumulated together. 
%

\vspace{-0.1in}
\subsection{Experiments}
\label{sec:exp}

In this section, we first discuss results from our experiments evaluating the different parameters affecting the accuracy-time trade-off.
We then report the performance of our technique when using the identified parameter values.
The experiments were performed on a workstation with a Intel Xeon E5-2620 CPU, 128~GB RAM, and an
Nvidia~GTX~1080 graphics card with 8 GB of GPU RAM.
We use Manhattan as the test bed for the experiments. The geometries of the buildings in the city were obtained through Open Street Maps and consist of over 43~thousand buildings present in Manhattan. The mesh is composed of 1.5~million triangles.

\myparagraph{Accuracy trade-off due to direction graph.}
A crucial step in improving the performance is grouping the set of possible light directions into a set of clusters, and using a representative of each cluster to approximate the light directions (Section~\ref{sec:large-time}).
As mentioned earlier, the maximum angle between any two directions in a cluster is bounded by a specified angle. A small angle, while having high accuracy will impede the efficiency. On the other hand, a large angle can drastically decrease the accuracy. 
We found that using a bound of $5^\circ$, we were able to get a good accuracy-time trade-off.
In particular, when testing $n = 1000$ random time steps, and computing the similarity between the actual direction, and the direction of the cluster representative, we found that the mean similarity measure was $0.9996$ with a standard deviation of $3 \times 10^{-4}$, implying that the clustering provides good approximation.

To quantify the effect this approximation has on the accumulated shadows, we chose a set of random camera positions and 1 hour time intervals, and computed the gross shadow when using both the actual direction (ground truth) and the cluster representative. 
Note that the gross shadow for this experiment was computed using the ray tracing baseline. The shadows were computed at a resolution of $800\times 600$.

Figure \ref{fig:plots}a plots the ground truth shadow area against the shadow area computed using the cluster representative. Recall that the shadow area is the weighted sum of pixel area, weighted by the gross shadow. 
The mean and median absolute error in the area was only $0.47\%$ and $0.35\%$ respectively, with a standard deviation of $0.37\%$, when compared to the total area. On an average, only $0.8\%$ of the points, with a standard deviation of $0.5\%$, were incorrectly tagged as being in or not in shadow.
%

%
\myparagraph{Maximum source level for accurate shadow accumulation.}
The primary use of inverse accrual maps is to accurately estimate the shadows for analysis. An important parameter affecting the accuracy of this approach is the maximum source level, $i$, that specifies the number of 2D slices of the inverse accrual map that is to be computed.
Recall that, given a time range $[t_1,t_n]$, the map is computed for both $t_1$ to $t_n$, as well as  $t_n$ to $t_1$.
To identify a suitable value, we chose a set of random hourly intervals and camera positions, and compared the computed gross shadow between the ray tracing baseline (ground truth) and inverse accrual maps by varying the maximum source level.

Figure \ref{fig:plots}b plots the average percentage error in the shadow area with increasing source levels. As expected, the error decreases with increasing number of levels. Using this plot, we fix the knee of this curve, \ie $i = 3$, as the parameter for computing inverse accrual maps. At this point, the mean error is less than $1\%$ of the total area. The median error for $i = 3$ is $0.7\%$, while the maximum error is $2.6\%$. 
We found that the maximum error occurred primarily when the accumulation was performed during dawn or dusk. This is because the shadows are not only long, but they also move quickly. In such a case it is possible for a single point to have several sources of shadow. Since we are considering only 3 sources, we miss considering shadows that are due to other sources.

\myparagraph{Performance evaluation.}
For the remainder of this paper, we use the parameters identified in the above experiments for computing shadow accrual maps as well as inverse accrual maps. 
We compare the effect of these parameters with both the baselines. 
These experiments consider the end-to-end time, which includes computing shadow accrual maps (or inverse accrual maps), and using them to compute and visualize the shadow accumulation.

For the first experiment, we consider 10 random positions, and 20 random days spread throughout the year. 
For each position-day pair, we compute the gross shadow for a period of 6 hours, starting from 9~am till 3~pm.
While the selected days cover the different seasons of the year, using a period of 6 hrs ensures that the different positions of the sun during the day are considered as well.
All shadows for this experiment were accumulated at a resolution of $800\times800$. Note that this was the output resolution. The shadow map and shadow accrual maps were at a resolution of $1024\times1024$.
Since inverse accrual maps compute shadow accumulation only along the ground plane, we modified the ray tracing baseline to also do the same.
Figure \ref{fig:perf}a plots the average time taken to compute shadow accumulation for an hour over the different days and camera positions. 
The reported time corresponds to the median computation time over 5 independent runs. 
On average, shadow accrual maps perform over 10X faster than the shadow map-based baseline, while inverse accrual maps perform around 5.3X faster than the ray tracing baseline. 

The second experiment tests the scalability of the approaches. We fixed a position and day, and computed the gross shadows for same period of 6 hours, but varying the resolution. For this experiment, we set the output resolution the same as the shadow map resolution.
Figure \ref{fig:perf}b plots the average time taken to render an hourly interval with increasing resolution. 
Note that both shadow accumulation approaches scale linearly with resolution.
%

\begin{table}[t]
\small
\centering
\caption{Memory required by the different approaches for varying resolutions. For shadow map baseline and shadow accrual map, the shadow map resolution is the same as the rendering resolution. 
}
\label{tab:mem}
\vspace{-0.15in}
\resizebox{\columnwidth}{!}{%
\begin{tabular}{|c|c|c|c|c|} \hline
\textbf{Resolution} & \textbf{Shadow Map} & \textbf{Shadow}& \textbf{Ray Tracing} & \textbf{Inverse} \\ 
 & \textbf{Baseline} & \textbf{Accrual Map}& \textbf{Baseline} & \textbf{Accrual Map} \\ \hline \hline
$512\times 512$		& 1~MB		& 60~MB		& 1~MB		& 14~MB 		\\
$800\times 800$		& 2.45~MB	& 146.5~MB	& 2.45~MB	& 34.2~MB	\\
$1024\times 1024$	& 4~MB		& 240~MB		& 4~MB		& 56~MB	\\
$1600\times 1600$	& 9.77~MB	& 586~MB		& 9.77~MB	& 136.7~MB	\\
$2048\times 2048$	& 16~MB	& 960~MB		& 16~MB	& 224~MB	\\ \hline
\end{tabular}}
\vspace{-0.25in}
\end{table}

\myparagraph{Performance improvement due to direction graph.}
When accumulating time periods involving multiple days, a brute force approach would explicitly compute shadows for every minute over all days. 
On the other hand, when using the direction graph there is a reuse of the shadow accrual maps (inverse accrual maps) across time steps having similar direction. This significantly reduces the number of shadow accrual map (inverse accrual map) computations.
Figure \ref{fig:perf}c plots the time taken to accumulate shadows over multiple days, accumulating for 6 hours each day. 
The advantage of the graph becomes apparent with increasing time periods. 
%
For example, when accumulating over a year, the brute force approaches would accumulate shadows for $365 \times 6 = 2190$ hourly intervals. 
Using the direction graph with a $5^\circ$ clustering bound, we only need to compute the maps corresponding to $299$ edges.

%


\begin{figure*}[t]
\centering
\includegraphics[width=\linewidth]{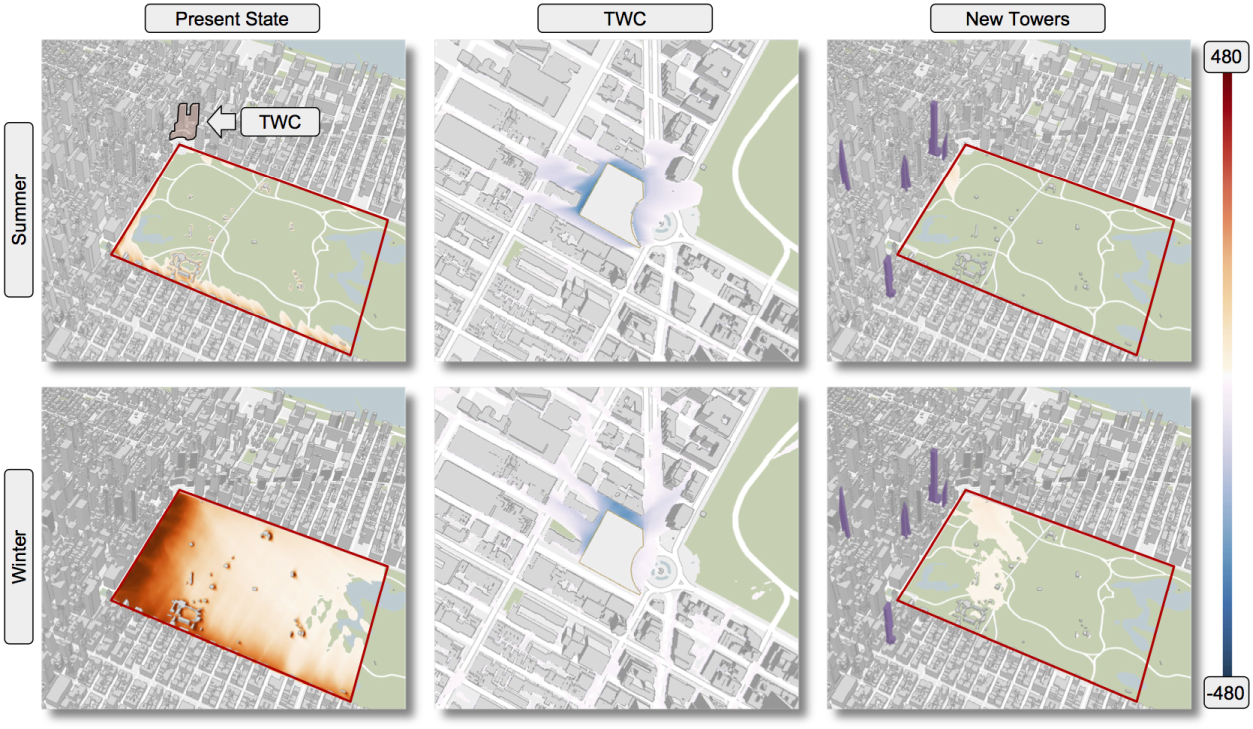}
\vspace{-0.2in}
\caption{Testing the impact of skyscrapers that are under construction south of Central park. Shadows cast during summer and winter with the current state (left). The impact of Time Warner Center~(TWC) is used as a baseline for comparison (middle). The impact of the new towers (right). Here, we are visualizing only regions having an impact (positive as well as negative) greater than 30~minutes.}
\label{fig:casestudy1}
\vspace{-0.2in}
\end{figure*}

\myparagraph{Memory requirements}
Since we are using $n = 60$, shadow accrual maps require storage equivalent of 60 shadow maps. Thus, when using a shadow map resolution of $1024\times1024$, 240~MB of GPU memory is used by shadow accrual maps (depth values stored as 4~byte floating points).
In case of inverse accrual maps, the additional memory required is directly proportional to the the number of source levels that are used. 
Recall that for a given point, storage is required for mapping the point along both $t_1$ to $t_n$ and $t_n$ to $t_1$ for each source level.
Since the points are on a plane, the mapped points can be represented using only 2 coordinates. 
Thus, with the number of source levels $l=3$, inverse accrual maps require 48~bytes of storage per pixel. Additionally, to maintain 
the accumulation, it also requires $n=60$ bits per pixel. Thus, when rendering shadows at a resolution of $1024\times1024$,
inverse accrual maps require approximately $56$~MB additional storage.
Table~\ref{tab:mem} lists the memory required by the different approaches when using different resolutions.

\begin{figure*}
\centering
\includegraphics[width=\linewidth]{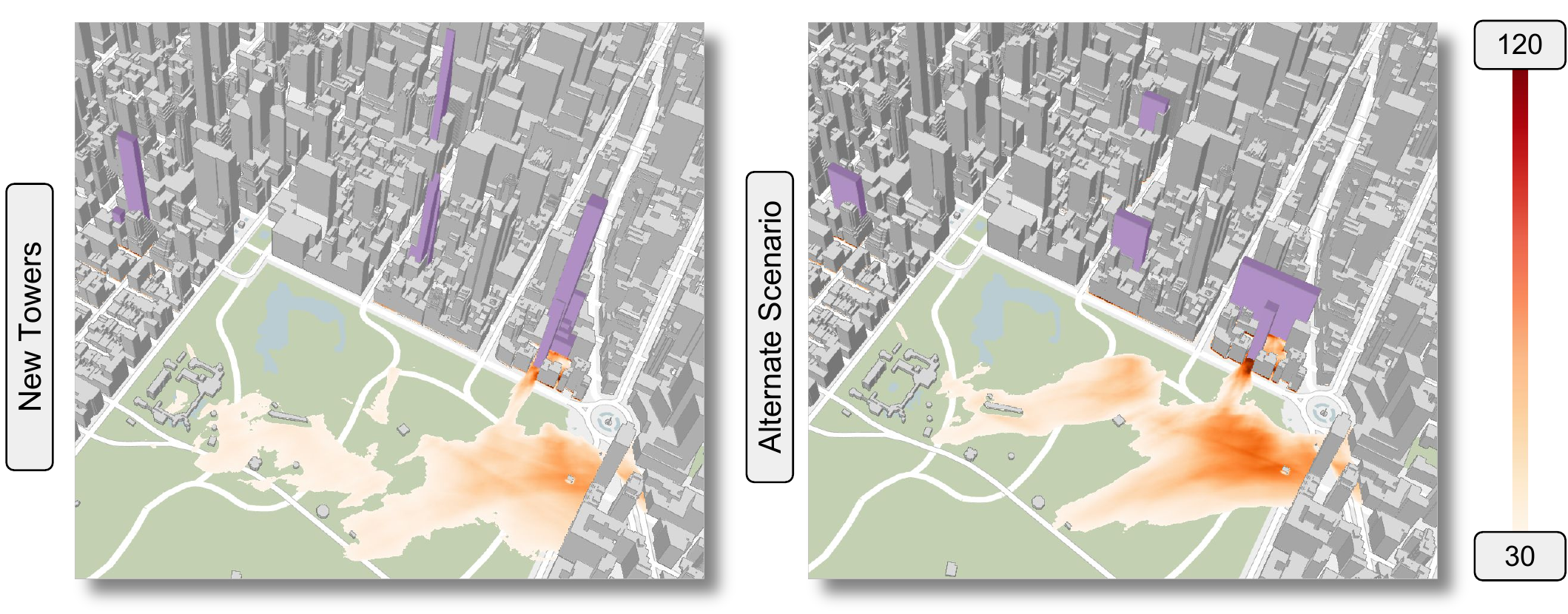}
\vspace{-0.3in}
\caption{We compare the impact of shadows from new and under construction set of skyscrapers in Manhattan (left) with an alternate scenario having shorter towers but with the same total area (right). The towers are highlighted in purple and the impact is visualized using a color map. Note that the impact is stronger due to the shorter towers.}
\label{fig:teaser}
\vspace{-0.2in}
\end{figure*}

\section{Case Studies}
\label{sec:case-studies}

In this section, we demonstrate the application of the Shadow Profiler system through two case studies in New York City~(NYC), a dense urban environment where the impact of new development on streets and parks is a constant concern. The first case study analyzes the impact of new development, more specifically skyscrapers, bordering on Central Park. The second compares various neighborhoods around NYC, specifically looking at desirable shade relative to determinable shadow.
Both case studies engage a variety of stakeholders from the general public and advocacy groups to government agencies, such as the City Council, the Department of Parks, and the Department of City Planning.

\subsection{Impact of buildings on Central Park}
\label{sec:central-park}

Just south of Central Park in Manhattan, a new generation of slender, supertall skyscrapers have begun to rise. There are seven skyscrapers recently built or under construction that range between 780~ft. and 1,490~ft. Their city-wide visibility and proximity to Central Park have raised concerns over shadows cast on the park. With this have come calls to revise NYC zoning regulations to include special review for new towers over 600~ft.~\cite{CentralParkSunshineTaskForce}. However, there has been little analysis of cast shadow done to test the impact. Whatever analysis that was done was only over fixed time instants~\cite{MASAccidentalSkyline}, which can be misleading since slender towers cast long shadows that move quickly. 
Comfort level for park-goers and impact on plant life is dependent on the duration of shadow. The longer a person is in shadow the cooler it gets; and plants need a certain number of hours of direct sunlight to grow. 
So shadows can be both beneficial as well as detrimental depending on the context.

\myparagraph{Impact of the proposed towers.}
In this study, we analyze the impact of the skyscrapers south of Central Park by differentiating between negative and beneficial shadows (shade) and using it to compare their performance with shorter and wider buildings. This can also inform a broader discussion on the development and regulation of supertall development in NYC. 
We divide our analysis time range into two periods representative of negative and beneficial shadows: November and December vs June and July. For consistency we consider an 8~hour period from 8:00~AM to 4:00~PM for each day. For these time periods we first analyze the shadows present in the current context without the seven new skyscrapers, shown in \reffig{fig:casestudy1}(left). The region of interest is highlighted in the figure and the gross shadow is visualized. We see that the shadows behave as expected -- a lower angle of the sun in winter causes shadows to cover the entire analysis area, while a higher angle in summer results in a tighter shadow area. 
Note that even though buildings, whether tall or short, cast long shadows at low sun angles (mornings and evenings), its contribution to the overall accumulation is small as reflected in the visualization.

As a baseline for comparison, we analyze the impact of Time Warner Center~(TWC) on Central Park. TWC is a skyscraper which is famous for having its design reworked after protests about the shadows it could cast~\cite{NYTColumbusCircle}. This is accomplished by removing the tower and computing the impact.
Even though the shadows due to TWC cover a large area, the actual area of shadow \textit{only because of TWC} is much smaller. Figure \ref{fig:casestudy1}(center) visualizes the region that is impacted by more than 30 minutes by TWC. Note that this is indeed a small region in its immediate neighborhood extending a little on the north of the building.

Next, we examine the impact of the seven new towers. Since the new cluster of tall towers are located in a wide area south of Central Park they affect a very large area of the park. However, when we restrict to areas impacted by an increase of over 30 minutes of shadows (\reffig{fig:casestudy1}(right)), we notice that this is a small fraction of the effected area. 
Note that this area while comparable for summer, is greater in winter than the concentrated impact that TWC had at Columbus Circle.

\myparagraph{Testing alternate scenarios.}
Finally, we use Shadow Profiler to test alternate development scenarios for the new skyscrapers. We modeled a new set of towers, all with the same area as the current seven, with larger floors and lower overall heights. For building lots that allowed it, we doubled the floor plate size. For others, we used the largest floor plate size that could be accommodated on the corresponding building lot. These allow us to test if height is indeed an issue that needs to be regulated. 
This resulted in shadows that are comparable in summer, but having a shorter spread in winter as shown in \reffig{fig:teaser}. However, notice that quantity of  impact (increase in gross shadow) is greater for the short but broader set, especially closer to their base.
In fact, when using Boston's shadow duration regulations and  considering the area impacted by greater that 60 minutes of new shadow, the shorter towers have greater impact compared to the proposed set.

Thus, choosing the right height with respect to shadow impact is essentially a trade-off between distribution and concentration. That is, given buildings of similar density, a taller building distributes its shadow further away with lower impact over that region, whereas a shorter one concentrates its impact over a smaller area. Given the contention of tall towers' effect on the southern portion of Central Park, this compromise of building height and shadow concentration is particularly important.

\begin{figure*}[t!]
\centering
\includegraphics[width=\linewidth]{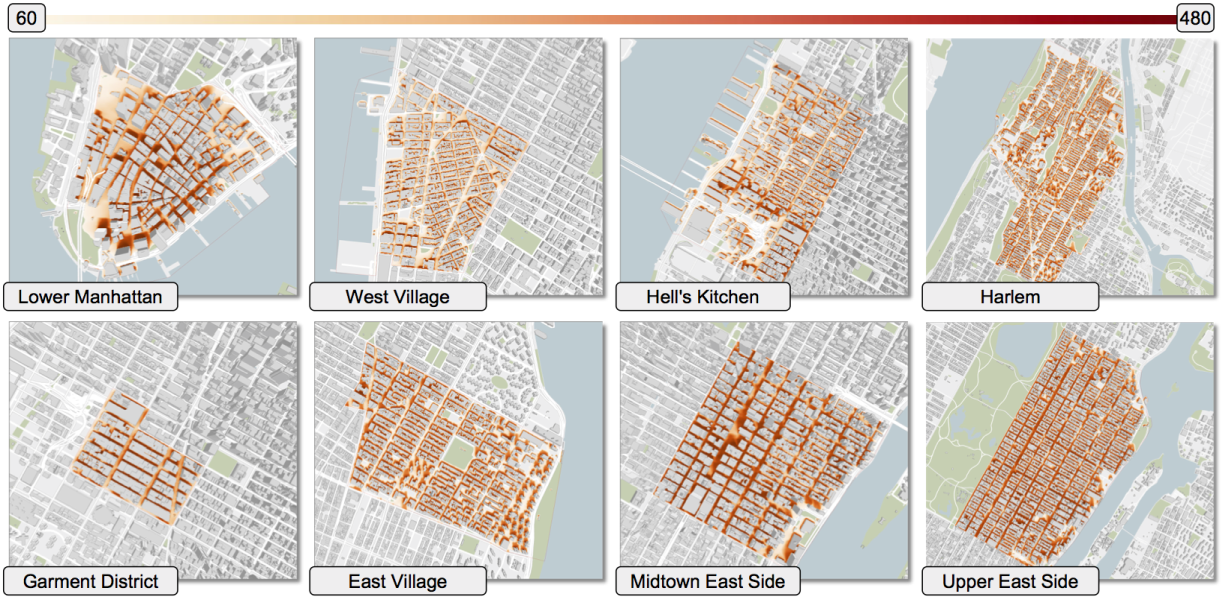}
\vspace{-0.25in}
\caption{Visually comparing year long gross shadows for different neighborhoods in Manhattan. The color map is set to visualize regions with gross shadows greater than an hour.}
\label{fig:use21}
\vspace{-0.2in}
\end{figure*}

\subsection{City Wide Shade Vs. Shadow}
\label{sec:shade-shadow}
City governments are tasked with preserving and promoting the quality of streets and public spaces. While daylight is protected as part of this, through zoning bulk regulations that dictate maximum buildings heights and setbacks, shadows are not rigorously controlled. As a result, a vast majority of new developments are never evaluated on shadows.
For projects that do warrant an evaluation, as mentioned earlier, it is only on a small scale primarily because existing tools are prohibitively expensive to scale up the analysis. 
Shadow Profiler, on the other hand, can allow city planners to analyze shadows comprehensively across the city and appropriately frame policy. 
%
%
In this case study, we first analyze the following neighborhoods for shadows over the entire year: Financial District, West Village, East Village, Garment District, Midtown East, Hell's Kitchen, Upper East Side, and Harlem. As before, the shadows are accumulated for 8 hours per day.
%
This comparison between neighborhoods, illustrated in \reffig{fig:use21}, reveals that most neighborhoods in Manhattan are in shadow for more than half the day on average over the entire year. This is expected given that Manhattan has a dense, heavily built urban grain. Closer examination of each neighborhood reveals that the concentration of shadows on streets and sidewalks correlates to the neighborhood's zoned density.
However, we find that wide streets, plazas, and parks have relatively lower shadows implying that such places are generally protected against excess shadows through controlling building densities and heights at these locations. For example, there is on an average less than an hour of gross shadow in the park in the center of East Village. 

\begin{table}[t]
\small
\centering
\vspace{-0.1cm}
\caption{Weights assigned to different months to characterize shade (desirable) vs shadow (undesirable).}
\label{tab:weights}
\vspace{-0.15in}
\resizebox{\columnwidth}{!}{%
\begin{tabular}{|c|c|c|c|c|c|c|c|c|c|c|c|} \hline
Jan & Feb & Mar & Apr & May & Jun & Jul & Aug & Sep & Oct & Nov & Dec \\ \hline \hline
-1 & -1 & -0.5 & 0 & 0.5 & 1 & 1 & 1 & 0.5 & 0 & -0.5 & -1 \\ \hline
\end{tabular}}
\vspace{-0.2in}
\end{table}

We next select three neighborhoods of interest in Manhattan -- West Village, Upper East Side, and Midtown, and analyze them to comprehensively understand how the built contexts relate to their experience through shade versus shadow. 
While city regulation identifies times when shadow is undesirable~\cite{CEQR}, given NYC's climate, shade produced during the hot summer months is also highly desirable. To make this distinction, we assign a positive weight for shade, \ie, shadows during summer, and a negative weight for shadows during winter. Table~\ref{tab:weights} shows the weights assigned to different months of the year. This weighing scheme is used to compute the shadow score of these neighborhoods over a period of one year.

Figure \ref{fig:use22} visualizes the overall score computed over the entire year for the three neighborhoods. Regions with a positive score are shades of blue, while those with a negative score are shades of red. The figure also shows the monthly distribution plots of shadow area~(orange plot) and shadow score~(blue plot). 
A region has an overall positive score if it is in shadow for a longer duration in summer than in winter. However, such locations only exist sporadically in lower building density regions of the neighborhoods. This basically indicates that buildings typically have a negative impact with respect to pedestrian comfort levels. 

\begin{figure*}[t!]
\centering
\includegraphics[width=\linewidth]{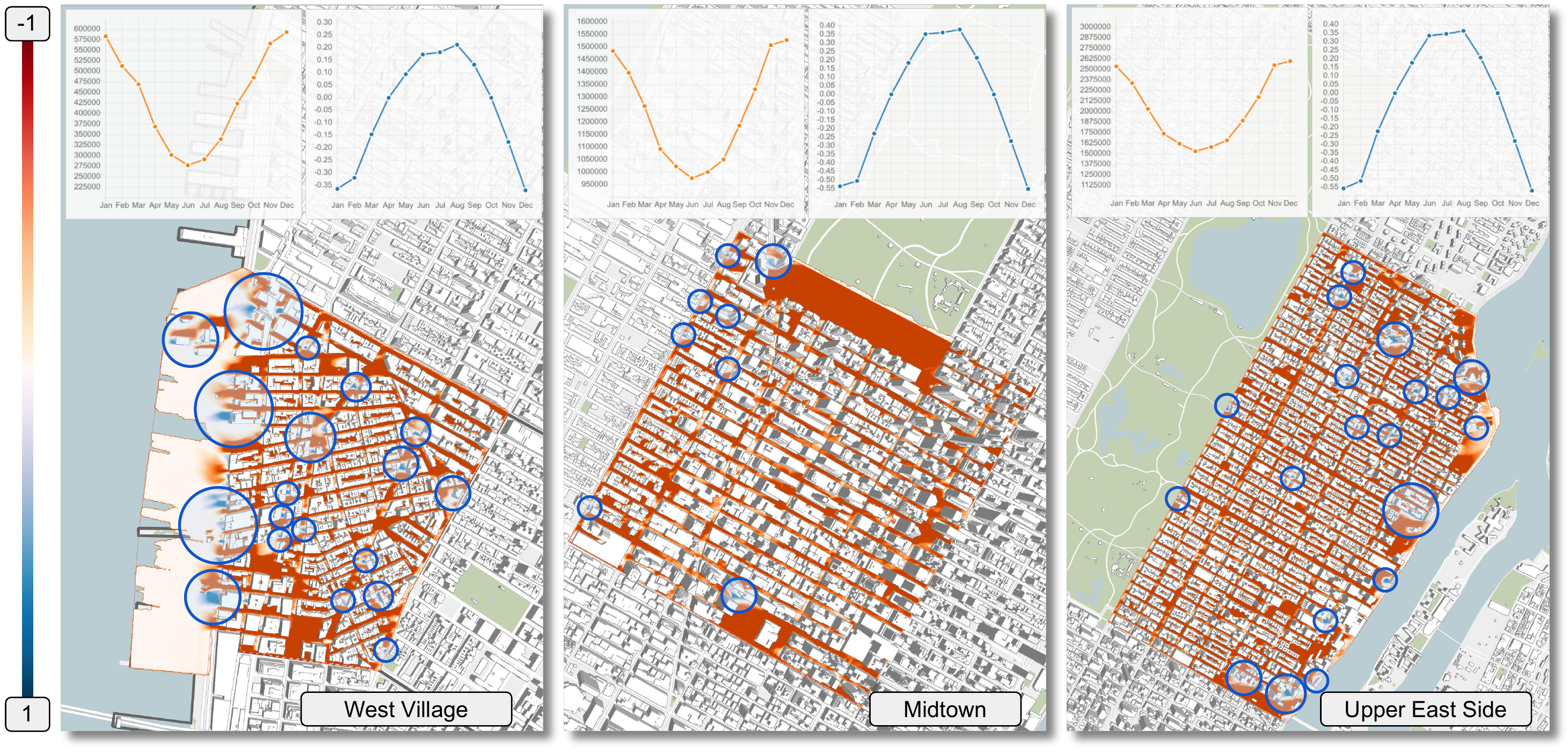}
\vspace{-0.2in}
\caption{The overall effect of shadows on 3 popular Manhattan neighborhoods is mostly negative. Regions with positive yearly score, highlighted by the blue circles, are sparsely distributed in low building density areas. Note that the areas of shadows typically decrease during summer months (orange plot) thus contributing less to the overall score.}
\label{fig:use22}
\vspace{-0.2in}
\end{figure*}

Looking back at \reffig{fig:casestudy1}, we see that such a behavior is true even for Central Park, where there is a higher concentration of shadows in winter than in summer. However, these situations for parks are generally mitigated by planting trees and other landscape features. More importantly, through analysis such as the one above it becomes possible to identify problematic regions and suggest corrective measures.
It can also be used by city planners to strategically incentivize new development for positive contributions to environmental quality as well as for designers to respond early to these objectives prior to civic review.

\section{Limitations and Future Work}
\myparagraph{Inverse accrual maps on arbitrary topography.}
Inverse accrual maps limit the use of computing accurate 
shadow measures to flat surfaces. 
This could be overcome through an hybrid approach---use inverse accrual maps for the ground, and resort to brute force ray tracing for buildings. 
While this can help in a city like NYC which is mostly flat, it will not be as helpful over arbitrary terrains.
We plan to explore other approaches, including the use of Monte Carlo ray tracing, to efficiently accumulate shadows in such situations. 

\myparagraph{Global illumination.}
Our current focus is on shadows due to direct sunlight.
Public spaces such as parks are typically large enough that the global illumination effects due to the facade materials of the buildings is minimal.
However, this will not hold when considering streets / plazas surrounded by several towers, many of which have a glass facade. 
In the future we plan to add functionality to support such scenarios.

\myparagraph{Extending other shadowing techniques.}
Any of the existing shadow maps or shadow volume based techniques can be extended to accumulate shadows by using the linear movement property to compute the shadows at intermediate time steps. 
This would primarily require computing the shadows for the first and last time steps within the given range (e.g. 1 hour), and 
appropriately using the interpolated values for the time steps in between.

\myparagraph{Conclusions.}
In this paper we proposed two techniques, shadow accrual maps and inverse accrual maps, 
to efficiently accumulate shadows over time. The key in our approach was to implicitly track shadows
based on the movement of the sun. 
We also reported experiments demonstrating the efficiency of our approaches.
These techniques were then used to develop
an interactive visual analysis tool called Shadow Profiler.
Through using Shadow Profiler to understand how different building types cast shadows, city planners can design zoning regulations to meet their goals while maximizing density and preserving public space quality. 
It can also function as a learning tool for the general public to understand the effect of shadows on cities~\cite{nyt-shadows}. 
We believe our framework is a first step to change the current planning practice by facilitating the transition from prescriptive rule based zoning to performance based zoning, and from discontinuous, isolated, and periodic nature of environmental review to functional continuous relationships between climate and city bulk regulations. 

\section*{Acknowledgments}
This work was supported in part by: the Moore-Sloan Data Science Environment at NYU; NASA; DOE; Kohn Pedersen Fox Associates; NSF awards CNS-1229185, CCF-1533564, CNS-1544753, CNS-1730396; CNPq; and FAPERJ. C.~T.~Silva is partially supported by the DARPA D3M program. Any opinions, findings, and conclusions or recommendations expressed in this material are those of the authors and do not necessarily reflect the views of DARPA.

\newpage
\bibliographystyle{IEEEtran}
\bibliography{IEEEabrv,paper}
\vspace*{-0.5in}
\begin{IEEEbiography}[{\raisebox{0.25in}{\includegraphics[width=1in,height=1.25in,clip,keepaspectratio]{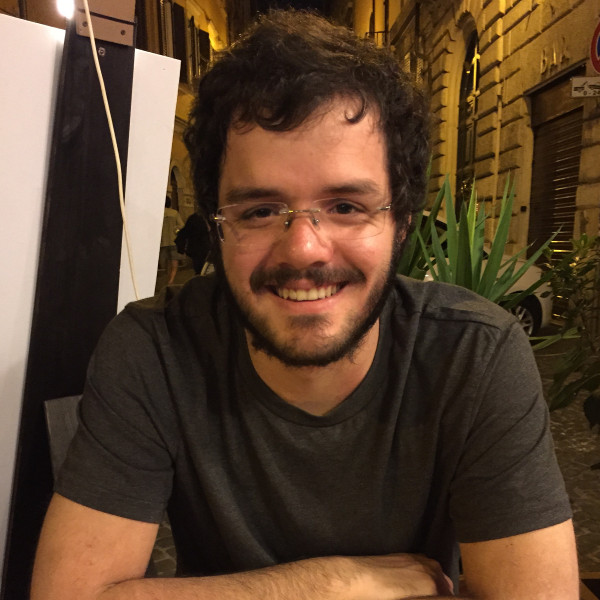}}}]{Fabio Miranda} is a Ph.D. candidate in the Computer Science and Engineering Dept. at NYU. He received the M.Sc. degree in computer science from PUC-Rio. During this period, he worked as a researcher and software engineer developing visualization tools for the oil industry. His research focuses on large scale data analysis, data structures, and urban data visualization.
\end{IEEEbiography}
\vspace*{-0.7in}
\begin{IEEEbiography}[{\raisebox{0.25in}{\includegraphics[width=1in,height=1.25in,clip,keepaspectratio]{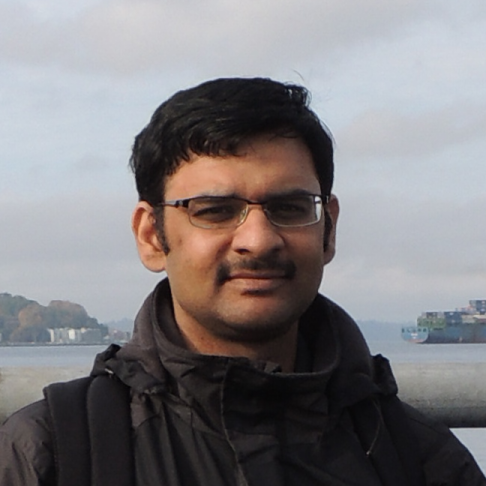}}}]{Harish Doraiswamy}
is a Research Scientist at the NYU Center of Data Science. He received his Ph.D. in Computer Science and Engineering from the Indian Institute of Science, Bangalore. His research interests lie in the intersection of computational topology, visualization, and data management. His recent research focuses on the analyses of large spatio-temporal datasets from urban environments.
\end{IEEEbiography}
\vspace*{-0.6in}
\begin{IEEEbiography}[{\raisebox{0.25in}{\includegraphics[width=1in,height=1.25in,clip,keepaspectratio]{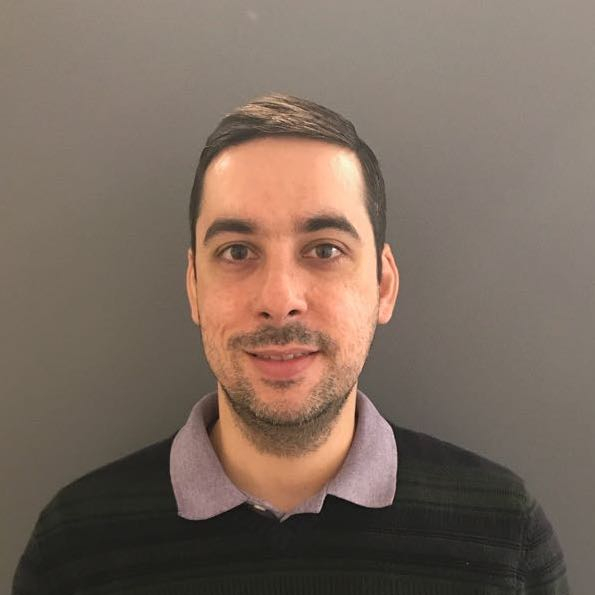}}}]{Marcos Lage}
is a professor in the Dept. of Computer Science at UFF, and is one of the principal investigators of the Prograf lab. His research interests include aspects of visual computing, especially scientific and information visualization, numerical simulations, geometry processing, and topological data structures. He has a Ph.D. in applied mathematics from PUC-Rio.
\end{IEEEbiography}

\vspace{-0.75in}

\begin{IEEEbiography}[{\raisebox{0.25in}{\includegraphics[width=1in,height=1.25in,clip,keepaspectratio]{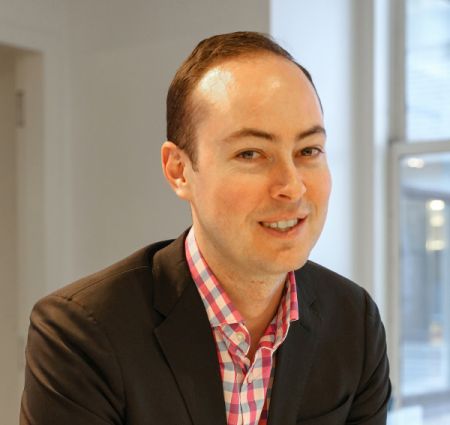}}}]{Luc Wilson}
is a senior associate principal at Kohn Pedersen Fox and the director of KPF Urban Interface, a think-tank focused on urban analytics and technology development. Luc is also an Adjunct Assistant Professor at Columbia GSAPP, and an Adjunct Course Advisor in the Center for Data Science at NYU. He earned his M.Arch. from Columbia University.
\end{IEEEbiography}

\vspace{-0.7in}

\begin{IEEEbiography}[{\raisebox{0.25in}{\includegraphics[width=1in,height=1.25in,clip,keepaspectratio]{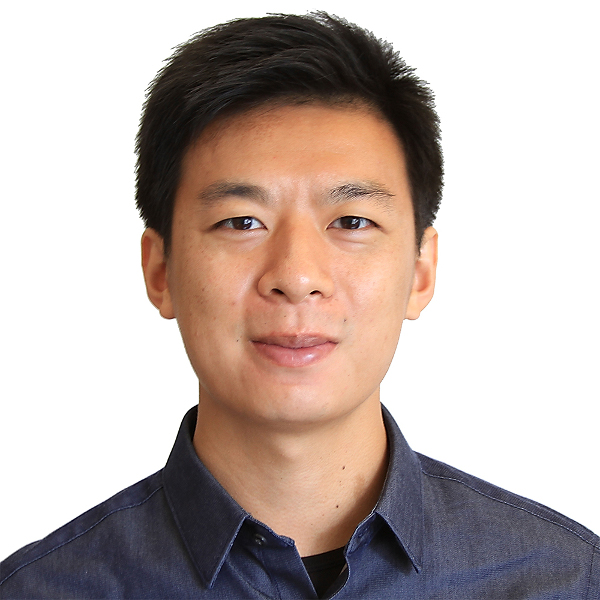}}}]{Mondrian Hsieh}
is a computational design at Kohn Pedersen Fox and researcher for KPFui. He has taught sustainable design and visual studies courses at Columbia School of Continuing Education and GSAPP. Mondrian Holds a B.A. in Architecture from UC Berkeley, and received his M.Arch. from Columbia University.
\end{IEEEbiography}

\vspace{-0.75in}

\begin{IEEEbiography}[{\raisebox{0.25in}{\includegraphics[width=1in,height=1.25in,clip,keepaspectratio]{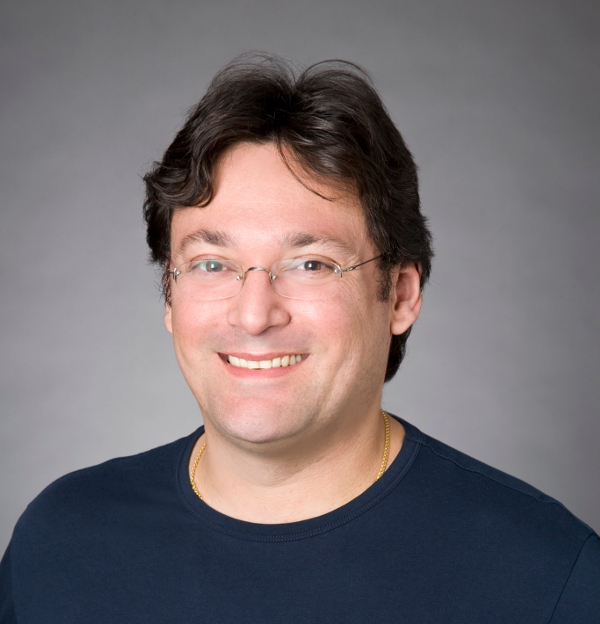}}}]{Cl\'{a}udio~T.~Silva}
is a professor of computer science and engineering and data science with NYU. His research lies in the intersection of visualization, data analysis, and geometric computing, and recently has focused on urban and sports data. He has received a number of awards: IEEE Fellow, IEEE Visualization Technical Achievement Award, and elected chair of the IEEE Visualization \& Graphics Technical Committee.
\end{IEEEbiography}

\end{document}